\documentclass[a4paper,numberwithinsect,UKenglish,cleveref,autoref,thm-restate,authorcolumns]{lipics-v2019}

\usepackage[utf8]{inputenc}

\usepackage{thmtools}

\usepackage{extarrows}
\usepackage{microtype}
\usepackage{mathtools}
\usepackage{amsmath}
\usepackage{amsthm}
\usepackage{thm-restate}

\usepackage{amssymb}
\usepackage{bbm}
\usepackage{stmaryrd}
\newcommand{\ldbrack}{\llbracket}
\newcommand{\rdbrack}{\rrbracket}
\usepackage{tikz}
\usetikzlibrary{patterns,arrows,arrows.meta,calc}
\usepackage[textwidth=2cm,textsize=small]{todonotes}
\usepackage{xspace}

\usepackage{wrapfig}
\usepackage{paralist}
\usepackage{enumitem}

\title{Timed games and deterministic separability}

\author{Lorenzo Clemente}{University of Warsaw, Poland}{clementelorenzo@gmail.com}{https://orcid.org/0000-0003-0578-9103}{Partially supported by the Polish NCN grant 2017/26/D/ST6/00201}
\author{Sławomir Lasota}{University of Warsaw, Poland}{sl@mimuw.edu.pl}{https://orcid.org/0000-0001-8674-4470}
{Partially supported by the ERC project Lipa (grant agreement No. 683080)}
\author{Radosław Piórkowski}{University of Warsaw, Poland}{r.piorkowski@mimuw.edu.pl}{https://orcid.org/0000-0002-9643-182X}
{Partially supported by the Polish NCN grant 2017/27/B/ST6/02093}

\authorrunning{Lorenzo Clemente, Sławomir Lasota, and Radosław Piórkowski} 
\Copyright{Lorenzo Clemente and Sławomir Lasota and Radosław Piórkowski} 



\category{} 

\relatedversion{} 

\supplement{}

\acknowledgements{We thank S. Krishna for fruitful discussions.
We kindly thank the anonymous reviewers for their helpful comments.}

\nolinenumbers 

\EventEditors{John Q. Open and Joan R. Access}
\EventNoEds{2}
\EventLongTitle{42nd Conference on Very Important Topics (CVIT 2016)}
\EventShortTitle{CVIT 2016}
\EventAcronym{CVIT}
\EventYear{2016}
\EventDate{December 24--27, 2016}
\EventLocation{Little Whinging, United Kingdom}
\EventLogo{}
\SeriesVolume{42}
\ArticleNo{23}

\newcommand{\proj}[1]{\mathsf{proj}(#1)}
\newcommand{\projinv}[1]{\mathsf{proj}^{-1}(#1)}

\newcommand{\dom}[1]{\textsf{dom}(#1)}
\newcommand{\one}[1]{\mathbf{0}(#1)}
\newcommand{\tick}{\boxempty} 
\newcommand{\para}[1]{\subparagraph*{#1.}}
\renewcommand{\thanks}{r_\emptyset}
\newcommand{\successor}[2]{\text{\sc succ}_{#1}(#2)}

\newcommand{\xsuccessor}[3]{\successor{#1}{#2, #3}}

\newcommand{\A}{\mathcal{A}}
\newcommand{\B}{\mathcal B}
\renewcommand{\S}{\mathcal{S}}

\newcommand{\startletter}{\rhd} 

\newcommand{\ignore}[1]{}

\newcommand{\goesto}[1]{\xrightarrow{#1}}

\newcommand{\lang}[1]{L(#1)}
\newcommand{\elang}[1]{L_\varepsilon(#1)}
\newcommand{\omegalang}[1]{L^\omega(#1)}
\newcommand{\omegaelang}[1]{L^\omega_\varepsilon(#1)}

\newcommand{\tr}{\mathsf{tr}}

\newcommand{\langts}[1]{\lang{#1}}         

\newcommand{\true}{\mathbf{true}}
\newcommand{\false}{\mathbf{false}}
\newcommand{\N}{\mathbb N}
\newcommand{\Z}{\mathbb Z}

\newcommand{\R}{\mathbb R}
\newcommand{\Rpos}{\R_{\geq 0}}

\newcommand{\X}{\mathtt X}
\newcommand{\Y}{\mathtt Y}

\newcommand{\T}{\mathtt T}
\newcommand{\x}{\mathtt x}
\renewcommand{\r}{\mathtt r}
\newcommand{\f}{\mathtt f}
\newcommand{\y}{\mathtt y}
\renewcommand{\L}{\mathtt L}
\newcommand{\Ii}{\mathtt I}
\newcommand{\F}{\mathtt F} 
\newcommand{\op}{\mathtt {op}}

\newcommand{\automaton}{\tuple{\Sigma, \L, \X, \Ii, \F, \Delta}}

\newcommand{\extend}[3]{#3[#1 \mapsto #2]}

\newcommand{\untime}[1]{\textsf{untime}(#1)}
\newcommand{\untimeinv}[1]{\textsf{untime}^{-1}(#1)}
\newcommand{\transition}[5]{\tuple{#1, #2, #3, #4, #5}}

\newcommand{\acc}{\ensuremath{\mathsf{acc}}}
\newcommand{\rej}{\ensuremath{\mathsf{rej}}}
\newcommand{\Acc}{{\mathsf{Acc}}}
\newcommand{\Rej}{{\mathsf{Rej}}}
\newcommand{\zeroregion}{\r_0}
\newcommand{\fract}[1]{\mathsf{fract}({#1})}
\newcommand{\req}[2]{\sim_{#1, #2}}
\newcommand{\class}[2]{[#2]_{#1}}
\newcommand{\region}[3]{\class {#1, #2} {#3}}
\newcommand{\regions}[2]{\textsf{Reg}(#1,#2)}
\newcommand{\fregions}[1]{\textsf{FReg}(#1)}

\newcommand{\OK}{\mathsf{OK}}
\newcommand{\ERROR}{\mathsf{ERROR}}

\newcommand{\DTA}{\textsf{\sc dta}\xspace}

\newcommand{\kDTA}[1]{\ensuremath{#1\textsf{-\sc dta}}\xspace}
\newcommand{\kmDTA}[2]{\ensuremath{#1, #2\textsf{-\sc dta}}\xspace}

\newcommand{\NTA}{\textsf{\sc nta}\xspace}
\newcommand{\kNTA}[1]{$#1$-\textsf{\sc nta}\xspace}
\newcommand{\kmNTA}[2]{$#1,#2$-$\textsf{\sc nta}$\xspace}

\newcommand{\NTAe}{\ensuremath{\textsf{\sc nta}^{\varepsilon}}\xspace}

\newcommand{\sem}[1]{\left\ldbrack#1\right\rdbrack}

\newcommand{\sep}{\ | \ }
\newcommand{\tuple}[1]{\left(#1\right)}

\newcommand{\set}[1]{\left\{ #1 \right\}}
\newcommand{\setof}[2]{\set{#1 \; \middle| \; #2}}

\newcommand{\LCM}{\textsf{LCM}\xspace}
\newcommand{\kLCM}[1]{$#1$-\textsf{LCM}\xspace}
\newcommand{\incr}[1]{#1\,\texttt{++}}
\newcommand{\decr}[1]{#1\,\texttt-\texttt-} 
\newcommand{\ztest}[1]{#1\stackrel ? {\texttt=} 0}

\newcommand{\reachset}[1]{\text{Reach}(#1)}

\newif\ifstartedinmathmode
\renewcommand*{\st}{
  \relax\ifmmode\startedinmathmodetrue\else\startedinmathmodefalse\fi
  \ifstartedinmathmode{\;\cdot\;}\else{s.t.~}\fi%
}

\newcommand{\wlg}{w.l.o.g.~}

\newcommand{\wrt}{w.r.t.~}

\newcommand{\card}[1]{|{#1}|}

\newcommand{\pspace}{{\sc PSpace}\xspace}

\newcommand{\game}[3]{G_{#1, #2}(#3)}
\newcommand{\M}{\mathcal M}
\newcommand{\Play}[1]{\mathsf{Play}(#1)}
\newcommand{\Run}[1]{\mathsf{Run}(#1)}
\newcommand{\iRun}[1]{\mathsf{Run_\omega}(#1)}

\newcommand{\I}{\text{I}\xspace}
\newcommand{\II}{\text{II}\xspace}
\newcommand{\rtop}[1]{\mathsf{r2p}(#1)}

\crefname{claim}{Claim}{Claims}
\Crefname{claim}{Claim}{Claims}

\crefname{lemma}{Lemma}{Lemmas}
\Crefname{Lemma}{Lemma}{Lemmas}

\crefname{theorem}{Theorem}{Theorems}
\Crefname{Theorem}{Theorem}{Theorems}

\crefname{fact}{Fact}{Facts}
\Crefname{fact}{Fact}{Facts}

\begin{document}

\maketitle

\begin{abstract}
	We study a generalisation of Büchi-Landweber games to the timed setting.
	The winning condition is specified by a non-deterministic timed automaton
	with epsilon transitions
	and only Player I can elapse time.
	We show that for fixed number of clocks and maximal numerical 
	constant available to Player II,
	it is decidable whether she has a winning timed controller using these
	resources.
	More interestingly, we also show that the problem remains decidable
	even when the maximal numerical constant is not specified in advance,
	which is an important technical novelty
	not present in previous literature on timed games.
	We complement these two decidability result
	by showing undecidability when the number of clocks available to Player II is not fixed.

	As an application of timed games, and our main motivation to study them,
	we show that they can be used to solve 
	the deterministic separability problem for nondeterministic timed automata with epsilon transitions.
	This is a novel decision problem about timed automata which has not been studied before.
	We show that separability is decidable when the number of clocks of the separating automaton is fixed
	and the maximal constant is not.
	The problem whether separability is decidable without bounding the number of clocks of the separator
	remains an interesting open problem.
\end{abstract}

%

\newpage


\section{Introduction}

\subparagraph{Separability.}

Separability is a classical problem in theoretical computer science and mathematics.
A set $S$ \emph{separates} two sets $L, M$
if $L \subseteq S$ and $S \cap M = \emptyset$.
%
Intuitively, a separator $S$ provides a certificate of disjointness,
yielding information on the structure of $L, M$ up to some granularity.
There are many elegant results in computer science and mathematics
showing that separators with certain properties always exist,
such as Lusin's separation theorem in topology
(two disjoint analytic sets are separable by a Borel set),
Craig's interpolation theorem in logic
(two contradictory first-order formulas can be separated by one containing only symbols in the shared vocabulary),
in model theory (two disjoint projective classes of models are separable by an elementary class),
in formal languages (two disjoint B\"uchi languages of infinite trees are separable by a weak language,
generalising Rabin's theorem \cite{Rabin:Weak:1970}),
in computability (two disjoint co-recursively enumerable sets are separable by a recursive set),
in the analysis of infinite-state systems
(two disjoint languages recognisable by well-structured transition systems are regular separable
\cite{CzerwinskiLasotaMeyerMuskallaKumarSaivasan:CONCUR:2018}), etc.

When separability is not trivial,
one may ask whether the problem is decidable.
Let $\mathcal C$ and $\mathcal S$ be two classes of sets.
The \emph{$\mathcal S$-separability} problem for $\mathcal C$
amounts to decide whether, for every input sets $L, M \in \mathcal C$
there is a set $S \in \mathcal S$ separating $L, M$.
Many results of this kind exist
when $\mathcal C$ is the class of regular languages of finite words over finite alphabets,
and $\mathcal S$ ranges over piecewise-testable languages
\cite{PlaceRooijenZeitoun:MFCS:2013,CzerwinskiMartensMasopust:ICALP:2013}
(later generalised to context-free languages \cite{CzerwinskiMartensVanRooijenZeitoun:FCT:2015}
and finite trees \cite{Goubault-Larrecq:Schmitz:ICALP:2016}),
locally and locally threshold testable languages \cite{PlaceRooijenZeitoun:LMCS:2014},
first-order logic definable languages \cite{PlaceZeitoun:LMCS:2016}
(generalised to some fixed levels of the first-order hierarchy \cite{PlaceZeitoun:ICALP:2014}).
For classes of languages $\mathcal C$ beyond the regular ones,
decidability results are more rare.
For example, regular separability of context-free languages is undecidable
\cite{SzymanksiWilliams:SIAM:1976,Hunt:JACM:1982,Kopczynski:LICS:2016}.
Nonetheless, there are positive decidability results for separability problems on several infinite-state models,
such as Petri nets \cite{ClementeCzerwinskiLasotaPaperman:STACS:2017},
Parikh automata \cite{ClementeCzerwinskiLasotaPaperman:Parikh:ICALP:2017},
one-counter automata \cite{CzerwinskiLasota:LMCS:2019},
higher-order and collapsible pushdown automata \cite{HagueKochemsOng:POPL:2016,ClementeParysSalvatiWalukiewicz:Diagonal:LICS16},
and others.

In this paper, we go beyond languages over finite alphabets,
and we study the separability problem for timed languages,
which we introduce next.

\subparagraph{Timed automata.}

Nondeterministic timed automata are one of the most widespread model of real-time reactive systems.
They consist of finite automata extended with real-valued clocks
which can be reset and compared by inequality constraints.
Alur and Dill's seminal result showed \pspace-completeness of the reachability problem \cite{AD94},
for which they received the 2016 Church Award \cite{church:award}.
This paved the way to the automatic verification of timed systems,
eventually leading to mature tools such as UPPAAL \cite{Behrmann:2006:UPPAAL4},
UPPAAL Tiga (timed games) \cite{CassezDavidFleuryLarsenLime:CONCUR:2005},
and PRISM (probabilistic timed automata) \cite{KwiatkowskaNormanParker:CAV:2011}.
%
The reachability problem is still a very active research area to these days
\cite{FJ15,
HerbreteauSrivathsanWalukiewicz:IC:2016,
AkshayGastinKrishna:LMCS:2018,
GastinMukherjeeSrivathsan:CONCUR:2018,
GastinMukherjeeSrivathsan:CAV:2019,
GovindHerbreteauSrivathsanWalukiewicz:CONCUR:2019},
as well as expressive generalisations thereof,
such as the binary reachability problem
\cite{ComonJurski:TA:1999,
Dima:Reach:TA:LICS02,
KrcalPelanek:TM:FSTTCS:2005,
FranzleQuaasShirmohammadiWorrell:IPL:2020}.

\emph{Deterministic timed automata} form a strict subclass of nondeterministic timed automata
where the next configuration is uniquely determined from the current one and the timed input symbol.
This class enjoys stronger properties,
such as decidable universality/inclusion problems and complementability \cite{AD94},
and it is used in several applications,
such as test generation \cite{NielsenSkou:STTT:2003},
fault diagnosis \cite{BouyerChevalierDSouza:FOSSACS:2005},
learning \cite{VerwerWeerdtWitteveen:Benelearn:2007,TapplerAichernigLarsenLorber:FOMATS:2019};
defining winning conditions in timed games
\cite{AsarinMaler:HSCC:1999,JurdzinskiTrivedi:ICALP:2007,BrihayeHenzingerPrabhuRaskin:ICALP:2007},
and in a notion of recognisability of timed languages \cite{Maler:Pnueli:FOSSACS:04}.

The \emph{$k,m$-deterministic separability} problem asks,
given two nondeterministic timed automata $\A$ and $\B$ with epsilon transitions,
whether there exists a deterministic timed automaton $\S$ with $k$ clocks and maximal constant bounded by $m$
\st $\lang \S$ separates $\lang \A, \lang \B$.
Likewise one defines \emph{$k$-deterministic separability}, where only $k$ is fixed but not $m$.
%
We can see $\A$ as recognising a set of good behaviours which we want to preserve
and $\B$ recognising a set of bad behaviours which we want to exclude;
a deterministic separator, when it exists,
provides a compromise between these two conflicting requirements.
%
%
To the best of our knowledge, separability problems for timed automata have not been investigated before.
%
%
Our first main result is decidability of $k, m$ and $k$-deterministic separability.
\begin{restatable}{theorem}{thm-km-separability}
    \label{sec:k:m:separability}
    \label{sec:k:separability}
    The $k, m$ and $k$-deterministic separability problems are decidable.
\end{restatable}%
%
\noindent
Decidability of deterministic separability should be contrasted with undecidability
of the corresponding membership problem 
\cite{Finkel:FORMATS:2006,Tripakis:IPL:2006}.
This is a rare circumstance,
which is shared with languages recognised by one-counter nets \cite{CzerwinskiLasota:LMCS:2019},
and conjectured to be the case for the full class of Petri net languages%
\footnote{
    All these classes of languages have a decidable disjointness problem,
    however regular separability is not always decidable in this case
    \cite{ThinniyamZetzsche:FSTTCS:2019}.
}.
We solve the separability problem by reducing to an appropriate timed game
(c.f.~\cref{thm:km:synthesis,thm:k:synthesis} below).
This forms the basis of our interest in defining and studying a non-trivial class of timed games,
which we introduce next.

\subparagraph{Timed games.}
We consider the following timed generalisation of Büchi-Landweber games \cite{BuchiLandweber:AMS:1969}.
There are two players, called Player \I and Player \II,
which play taking turns in a strictly alternating fashion.
At the $i$-th round, Player \I selects a letter $a_i$ from a finite alphabet and a nonnegative timestamp $t_i$ from $\Rpos$,
and Player \II replies with a letter $b_i$ from a finite alphabet.
At doomsday, the two players have built an infinite play
$\pi = \tuple{a_1, b_1, t_1} \tuple{a_2, b_2, t_2} \cdots$,
and Player \I wins if, and only if, $\pi$ belongs to her winning set, 
which is a timed langauge recognised by a nondeterministic timed automaton with $\varepsilon$-steps.
For a fixed number of clocks $k \in \N$ and maximal constant $m \in \N$,
the \emph{$k, m$-timed synthesis problem} asks whether there is a finite-memory timed controller for Player \II
using at most $k$ clocks and guards with maximal constant bounded by $m$ in absolute value,
ensuring that every play $\pi$ conform to the controller is winning for Player \II.
%
Our second contribution is decidability of this problem.
\begin{restatable}{theorem}{thm-km-synthesis}
    \label{thm:km:synthesis}
    For every fixed $k, m \in \N$, the $k, m$-timed synthesis problem is decidable.
\end{restatable}
\noindent
We reduce to an untimed finite-state game
with an $\omega$-regular winning condition
\cite{BuchiLandweber:AMS:1969}.
This should be contrasted with undecidability of the same problem
when the set of winning plays for Player \II
is a nondeterministic timed language
(c.f.~\cite{DsouzaMadhusudan:STACS:2002} for a similar observation).
%
%
The \emph{$k$-timed synthesis problem} asks whether there exists a bound $m \in \N$
\st the $k, m$-timed synthesis problem has a positive answer for Player \II,
which we also show decidable.
%
%
\begin{theorem}
    \label{thm:k:synthesis}
    For every fixed $k \in \N$, 
    the $k$-timed synthesis problem is decidable.
\end{theorem}
This requires the synthesis of the maximal constant $m$,
which is a very interesting a technical novelty
not shared with the current literature on timed games. 
%
%
We design a protocol whereby Player \II demands Player \I to be informed
when clocks elapse one time unit.
We require that the number of such consecutive requests be finite,
yielding a bound on $m$ (when such a value exists).

Finally, we complement the two decidability results above
by showing that the synthesis problem is undecidable
when the number of clocks $k$ available to Player \II is not specified in advance
(c.f.~\Cref{thm:synUnd}).


There are many variants of timed games in the literature,
depending whether the players must enforce a nonzeno play,
who controls the elapse of time,
concurrent actions, etc.
\cite{Wong-Toi-Hoffmann:CDC:1991,
MalelPnueliSifakis:STACS:1995,
AsarianMalerPnueliSifakis:SSSC:1998,
DsouzaMadhusudan:STACS:2002,
deAlfaroFaellaHenzingerMajumdarStoelinga:CONCUR:2003}.
%
In this terminology, our timed games are asymmetric (only Player \I can elapse time)
and turn-based (the two players strictly alternate).

\section{Preliminaries}

Let $\R$ be the set of real numbers and $\Rpos$ the set of nonnegative real numbers.
For two sets $A$ and $B$, let their Cartesian product be $A \cdot B$.
Let $A^0 = \set{\varepsilon}$,
and, for every $n \geq 0$, $A^{n+1} = A \cdot A^n$.
The set of finite sequences over $A$ is $A^* = \bigcup_{n \geq 0} A^n$,
$A^\omega$ is the set of infinite sequences,
and $A^\infty = A^* \cup A^\omega$.
A \emph{(monotonic) timed word} over a finite alphabet $\Sigma$ is a sequence
$w = \tuple{a_1, t_1} \tuple{a_2, t_2} \cdots \in (\Sigma \cdot \Rpos)^\infty$
\st $0 \leq t_1 \leq t_2 \leq \cdots$,
and it is \emph{strictly monotonic} if $0 \leq t_1 < t_2 < \cdots$.
A \emph{timed language} over $\Sigma$
is a set $L \subseteq (\Sigma \cdot \Rpos)^\infty$ of monotonic timed words;
it is \emph{strictly monotonic} if it contains only strictly monotonic timed words.
The \emph{untiming} $\untime w$ of a timed word $w$ as above
is the word $a_0 a_1 \cdots \in \Sigma^\infty$
obtained from $w$ by removing the timestamps,
which is extended to timed languages $L$ pointwise as
$\untime L = \setof {\untime w} {w \in L}$.

\para{Clocks, constraints, and regions}

Let $\X = \set{\x_1, \dots, \x_k}$ be a finite set of clocks.
A \emph{clock valuation} is a function $\mu \in \Rpos^\X$
assigning a nonnegative real number $\mu(\x)$ to every clock $\x \in \X$.
For a nonnegative time elapse $\delta \in \Rpos$,
we denote by $\mu + \delta$ the valuation assigning $\mu(\x) + \delta$ to every clock $\x$;
for a set of clocks $\Y \subseteq \X$, let $\extend \Y 0 \mu$
be the valuation which is $0$ on $\Y$ and agrees with $\mu$ on $\X \setminus \Y$.
We write $\mu_0$ for the clock valuation
mapping every clock $\x \in \X$ to $\mu_0(\x) = 0$.
A \emph{clock constraint} is a quantifier-free formula of the form
\begin{align*}
    \varphi,\psi \;::\equiv\; \true \sep \false \sep \x_i - \x_j \sim z \sep \x_i \sim z \sep \neg \varphi \sep \varphi \land \psi \sep \varphi \lor \psi,
\end{align*}
where $\sim \in \set{=, <, \leq, >, \geq}$
and $z \in \Z$.
A clock valuation $\mu$ satisfies a constraint $\varphi$, written $\mu \models \varphi$,
if interpreting each clock $\x_i$ by $\mu(\x_i)$ makes $\varphi$ true.
A constraint $\varphi$ defines the set
$\sem{\varphi} = \setof{\mu \in \Rpos^\X}{\mu \models \varphi}$
of all clock valuation it satisfies.
%
%
When the set of clocks is fixed to $\X$ 
and the absolute value of constants is bounded by $m \in \N$,
we speak of $\X, m$-constraints.
%
%
Two valuations $\mu, \nu \in \Rpos^\X$ are \emph{$\X, m$-region equivalent}, written $\mu \req \X m \nu$,
if they satisfy the same $\X, m$-constraints.
An $\X,m$-region $\region \X m \mu \subseteq \Rpos^\X$
is an equivalence class of clock valuations \wrt $\req \X m$.
%
%
For fixed finite $\X$ and $m \in \N$ there are finitely many $\X, m$-regions;
let $\regions \X m$ denote this set.
Let $\mu_0 = \lambda \x . 0$ and $\zeroregion = \region \X m {\mu_0}$ be its region.
%
We write $\r \models \varphi$ for a region $\r \in \regions \X m$
whenever $\mu \models \varphi$ for some $\mu \in \r$ (equivalently, for all such $\mu$'s).
The \emph{characteristic clock constraint} $\varphi_\r$
of a region $\r \in \regions \X m$
is the unique constraint (up to logical equivalence) \st $\sem{\varphi_\r} = \r$.
When convenient, we deliberately confuse regions with their characteristic constraints.
For two regions $\r, \r' \in \regions \X m$
we write $\r \prec \r'$ whenever $\r = \region \X m \mu$,
$\r' = \region \X m {\mu + \delta}$ for some $\delta > 0$,
and $\r \neq \r'$.
%
%
%

%
%

\para{Timed automata}
A (nondeterministic) \emph{timed automaton} is a tuple $\A = \automaton$,
where $\Sigma$ is a finite input alphabet,
$\L$ is a finite set of control locations,
$\X$ is a finite set of clocks,
$\Ii, \F \subseteq \L$ are the subsets of initial, resp., final, control locations,
and $\Delta$ is a finite set of transition rules of the form 
%
	$\tr = \transition p a \varphi \Y q \in \Delta$,
%
with $p, q \in \L$ control locations, $a \in \Sigma_\varepsilon := \Sigma \cup \set \varepsilon$,
$\varphi$ a clock constraint to be tested and $\Y \subseteq \X$ the set of clocks to be reset to $0$.
%
%
A \emph{configuration} of a timed automaton $\A$ is a pair $\tuple{p, \mu}$
consisting of a control location $p \in \L$
and a clock valuation $\mu \in \Rpos^\X$.
It is \emph{initial} if $p$ is so and $\mu = \mu_0$.
It is \emph{final} if $p$ is so.
Every transition rule $\tr$ 
induces a discrete transition between configurations
$\tuple {p, \mu} \goesto \tr \tuple {q, \nu}$ when
$\mu \models \varphi$ and $\nu = \extend \Y 0 \mu$.
Intuitively, a discrete transition consists of a test of the clock constraint $\varphi$, reset of clocks $\Y$, and step to the location $q$.
Moreover, for every nonnegative $\delta \in \Rpos$ and every configuration $\tuple{p, \mu}$
there is a time-elapse transition
$\tuple {p, \mu} \goesto \delta \tuple {p, \mu + \delta}$.
The timed language \emph{$\varepsilon$-recognised} by $\A$,
denoted $\elang \A$,
is the set of finite timed words
$w = \tuple {a_1, t_1} \cdots \tuple{a_n, t_n} \in (\Sigma_\varepsilon \cdot \Rpos)^*$
\st there is a sequence of transitions
$\tuple{p_0, \mu_0} \goesto {\tr_1, \delta_1} \cdots \goesto {\tr_n, \delta_n} \tuple{p_n, \mu_n}$
where $p_0 \in I$ is initial,
$\mu_0(\x) = 0$ for every clock $\x \in \X$,
$p_n \in F$ is final,
and, for every $1 \leq i \leq n$,
$\delta_i = t_i - t_{i-1}$ (where $t_0 = 0$)
and $\tr_i$ is of the form $\transition {p_{i-1}} {a_i} {\_} {\_} {p_i}$.
The timed $\omega$-language $\omegaelang \A \subseteq (\Sigma_\varepsilon \cdot \Rpos)^\omega$
is defined in terms of sequences as above with the condition that $p_i \in F$ infinitely often.
We obtain the timed language
$\lang \A = \pi(\elang \A) \subseteq (\Sigma \cdot \Rpos)^*$,
resp., $\omega$-language
$\omegalang \A = \pi(\omegaelang \A) \cap (\Sigma \cdot \Rpos)^\omega \subseteq (\Sigma \cdot \Rpos)^\omega$ \emph{recognised} by $\A$,
where $\pi$ is the mapping that removes letters of the form $\tuple{\varepsilon, \_}$.

A timed automaton (without $\varepsilon$-transitions) is \emph{deterministic}
if it has exactly one initial location 
and, for every two rules 
$\transition{p}{a}{\varphi}{\Y}{q}$, $\transition{p}{a}{\varphi'}{\Y'}{q'}$ with $\sem{\varphi \land \varphi'} \neq \emptyset$,
we have $\Y = \Y'$ and $q = q'$.
%
We write \NTA, \DTA for the classes of nondeterministic, resp., deterministic timed automata
without epsilon transitions.
When the number of clocks in $\X$ is bounded by $k$ we write \kNTA k, resp., \kDTA k.
When the absolute value of the maximal constant is additionally bounded by $m \in \N$
we write \kmNTA k m, resp., \kmDTA k m.
When epsilon transitions are allowed, we write \NTAe.
A timed language is called \NTA language, \DTA language, and so on,
if it is recognized by a timed automaton in the respective class.
A \kmDTA{k}{m} with clocks $\X$ is \emph{regionised}
if each constraint is a characteristic constraint $\varphi_\r$ of some region $\r \in \regions \X m$
and for each location $p$, input $a\in\Sigma$, and $r\in \regions \X m$ there is a (necessarily unique) transition rule of the form $\transition p a {\varphi_\r} \Y q$.
It is well-known that a \kmDTA{k}{m} can be transformed into an equivalent regionised one
by adding exponentially many transitions.

%

\begin{example}[\NTA language which is not a \DTA language] \label{example:L}
	Let $\Sigma = \set a$ be a unary alphabet
	and let $L$ be the set of timed words of the form
	$(a, t_1) \cdots (a, t_n)$
	\st $t_n - t_i = 1$ for some $1\leq i < n$.
	$L = \langts A$ for the timed automaton $\A = \automaton$
	with a single clock $\X = \set \x$ three locations $\L = \set{p, q, r}$,
	of which $\Ii = \set p$ is initial and $\F = \set r$ is final,
	and transitions rules $\transition{p}{a}{\true}{\emptyset}{p}$,
	$\transition{p}{a}{\true}{\set{\x}}{q}$,
	$\transition{q}{a}{\x<1}{\emptyset}{q}$,
	$\transition{q}{a}{\x=1}{\emptyset}{r} \in \Delta$.
	Intuitively, in $p$ the automaton waits until it guesses that the next input will be $(a, t_i)$,
	at which point it moves to $q$ by resetting the clock (and subsequently reading $a$).
	From $q$, the automaton can accept by going to $r$ only if exactly one time unit elapsed since $(a, t_i)$.	
%
	There is no \DTA recognising $L$,
	since in order to recognise $L$ deterministically one must store all timestamps in the last unit interval,
	and thus no bounded number of clocks suffices.
\end{example}
\begin{example}\label{example:M}
	The complement of $L$ from \cref{example:L}
	can be recognised by an \NTA with two clocks.
	Indeed, a timed word $(a, t_1) \cdots (a, t_n)$ is not in $L$
	if either of the following conditions hold:
	\begin{inparaenum}[1)]
		\item its length $n$ is at most $1$, or
		\item the total time elapsed between the first and the last letter is less than one time unit $t_n - t_1 < 1$, or
		\item there is a position $1 \leq i < n$
		\st $t_n - t_i > 1$ and $t_n - t_{i+1} < 1$.
	\end{inparaenum}
	It is easy to see that two clocks suffice to nondeterministically check the conditions above.
\end{example}
\noindent
Since checking whether an \NTA recognises a deterministic language is undecidable
\cite{Finkel:FORMATS:2006,Tripakis:IPL:2006},
there is no recursive bound on the number of clocks sufficient to deterministically recognise an \NTA language
(whenever possible).
Thus \NTA can be non-recursively more succinct than \DTA w.r.t.~number of clocks.
%
However, in general such \NTA recognise timed languages whose complement is not an \NTA language.
The next example shows a timed language which is both \NTA and co-\NTA recognisable,
however the number of clocks of an equivalent \DTA is at least exponential in the number of clocks of the \NTA.
\begin{example} \label{example:Lk}
	For $k \in \N$, let $L_k$ be the set of strictly monotonic timed words
	$(a, t_1) \cdots (a, t_n)$
	s.t.~$t_n - t_i = 1$ where $i = n - 2^k$.
	The language $L_k$ can be recognised by a $(2 \cdot k + 2)$-clock \NTA $\A_k$ of polynomial size.
	There are clocks $\x_0, \x_1, \dots, \x_k$ and $\y_0, \y_1, \dots, \y_k$.
	Clock $\x_0$ is used to check strict monotonicity.
	Clock $\y_0$ is reset when the automaton guesses $(a, t_i)$.
	The automaton additionally keeps track of the length of the remaining input.
	This is achieved by implementing a $k$-bit binary counter,
	where $\x_j = \y_j$ represents that the $j$-th bit is one.
	In order to set the $j$-th bit to one, the automaton resets $\x_j, \y_j$;
	to set it to zero, it resets only $\x_j$.
	This is correct thanks to strict monotonicity.
	At the end the automaton checks $\y_0 = 1$ and that the binary counter has value $2^k$.
	Any deterministic automaton recognising $L_k$ requires exponentially many clocks
	to store the last $2^k$ timestamps.
	The complement of $L_k$ can be recognised by a $(2 \cdot k + 2)$-clock \NTA of polynomial size.
	Indeed, a timed word is not in $L_k$ if any of the following conditions hold:
	\begin{inparaenum}[1)]
		\item its length $n$ is $\leq 2^k$, or
		\item $t_n - t_i < 1$ with $i = n-2^k$, or
		\item $t_n - t_i > 1$ with $i = n-2^k$.
	\end{inparaenum}
	The automaton guesses which condition holds
	and uses a $k$-bit binary counter as above to check that position $i$ has been guessed correctly.
\end{example}

%


\section{Timed synthesis games}
\label{sec:synthesis}

Let $A$ and $B$ be two finite alphabets of actions
and let $W \subseteq (A \cdot B \cdot \Rpos)^\omega$
be a language of timed $\omega$-words over the alphabet $A \cdot B$.
The \emph{timed synthesis game} $\game A B W$ is played by Player \I and Player \II in rounds.
At round $i \geq 0$,
Player \I chooses a timed action $a_i \cdot t_i \in A \cdot \Rpos$
and Player \II replies immediately with an untimed action  $b_i \in B$.
The game is played for $\omega$ rounds,
and at doomsday the two players have produced an infinite play
\begin{align}
    \label{eq:playinit}
    \pi = a_1 b_1 t_1 a_2 b_2 t_2 \cdots \in (A \cdot B \cdot \Rpos)^\omega.
\end{align}
Player \I wins the game if, and only if, $\pi \in W$.

Let $k \in \N$ be a bound on the number of available clocks $\X = \set{\x_1, \dots, \x_k}$,
and let $m \in \N$ be a bound on the maximal constant.
A \emph{$k,m$-controller} for Player \II in $\game A B W$
is a regionised \kmDTA{k}{m} $\M = \tuple{A, B, \L, \ell_0, \delta}$
with input alphabet $A$ and output alphabet $B$,
where $\L$ is a set of memory locations,
$\ell_0 \in \L$ is the initial memory location,
and $\delta : \L \cdot A \cdot \regions k m \to \L \cdot B \cdot 2^\X$ is the update function
mapping the current memory $\ell \in \L$,
input $a \in A$,
and region $\varphi \in \regions k m$
to $\delta(\ell, a, \varphi) = \tuple{\ell', b, \Y}$,
where 
$\ell' \in \L$ is the next memory location, $b \in B$ is an output symbol,
and $\Y \subseteq \X$ is the set of clocks to be reset.

We define by mutual induction the notion of \emph{$\M$-conform partial runs}
$\Run \M \subseteq \L \cdot \Rpos^\X \cdot (A \cdot B \cdot \Rpos \cdot \L \cdot \Rpos^\X)^*$,
and the strategy $\sem \M : \Run \M \cdot A \cdot \Rpos \to \L \cdot \Rpos^\X \cdot B$
induced by the controller on conform runs as follows:
Initially, $\tuple{\ell_0, \mu_0} \in \Run \M$,
where $\mu_0(\x) = 0$ for every clock $\x \in \X$.
Inductively,
for every $n \geq 0$ and every $\M$-conform partial run
\begin{align}
    \label{eq:conform:run}
    \rho =
        \tuple{\ell_0, \mu_0} 
        \tuple{a_1, b_1, t_1, \ell_1, \mu_1} \cdots
        \tuple{a_n, b_n, t_n, \ell_n, \mu_n} \in \Run \M,
\end{align}
and for every $\tuple{a_{n+1}, t_{n+1}} \in A \cdot \Rpos$, we define
$\sem \M (\rho \cdot a_{n+1} \cdot t_{n+1}) = \tuple{\ell_{n+1}, \mu_{n+1}, b_{n+1}}$
for the unique $(\ell_{n+1}, \mu_{n+1}, b_{n+1}) \in \L \cdot \Rpos^\X \cdot B$ \st 
$\delta(\ell_n, a_{n+1}, \varphi_{\mu_n + \delta_{n+1}}) = \tuple{\ell_{n+1}, b_{n+1}, \Y}$ and
${\mu_{n+1} = \extend \Y 0 {(\mu_n + \delta_{n+1})}}$,
where $\delta_{n+1} = t_{n+1} - t_n$ (with $t_0 = 0 $).
Moreover, $\rho \cdot a_{n+1} \cdot b_{n+1} \cdot t_{n+1} \cdot \ell_{n+1} \cdot \mu_{n+1} \in \Run \M$.
%
%
An infinite $\M$-conform run is any sequence 
$\rho \in \L \cdot \Rpos^\X \cdot (A \cdot B \cdot \Rpos \cdot \L \cdot \Rpos^\X)^\omega$
such that every finite prefix thereof is $\M$-conform;
let $\iRun \M$ be the set of such $\rho$'s.
Let ${\rtop \rho\in (A \cdot B \cdot \Rpos)^\omega}$
be 
the corresponding play $\pi = \rtop \rho$ as in \eqref{eq:playinit} obtained by dropping locations and clocks valuations.
The controller $\M$ is \emph{winning} if every infinite $\M$-conform run $\rho$
satisfies $\rtop \rho \not\in W$.
%
%
A \emph{$k$-controller} is $k, m$-controller for some $m\in\N$.
%
For fixed $k, m \in \N$,
the \emph{$k, m$-timed synthesis problem}
asks, given $A, B$ and an \NTAe timed language $W\subseteq (A\cdot B\cdot \Rpos)^\omega$, 
whether Player \II has a winning $k, m$-controller in $\game A B W$;
the \emph{$k$-timed synthesis problem} asks instead for a $k$-controller;
finally, the \emph{timed synthesis problem} asks whether there exists a controller.
%
The $0, 0$-timed synthesis problem is equivalent to untimed synthesis problem,
which is decidable by the B\"uchi-Landweber Theorem
\protect{\cite[Theorem $1'$]{BuchiLandweber:AMS:1969}}:
\begin{restatable}{lemma}{lemZeroSynthesis}
    \label{lem:zero:synthesis}
    The $0, 0$-synthesis problem is decidable.
\end{restatable}
%
%

\section{Deterministic separability}

In this section we prove our first main result \Cref{sec:k:m:separability}:
we show that the $k, m$ and $k$-deterministic separability problems are decidable.
We begin with a motivating example of nonseparable languages.
\begin{example}
    \label{example:nonseparable}
    Consider the \NTA language $L$ from \cref{example:L}.
    Thanks to \cref{example:M} its complement is also a \NTA language.
    Since neither $L$ nor its complement are deterministic,
    they cannot be deterministically separable.
\end{example}
\noindent
Moreover, a deterministic separator, when it exists,
may need exponentially many clocks.
\begin{example}
    \label{example:separable:exponential}
    We have seen in \cref{example:Lk} an $O(k)$-clock \NTA language
    \st 1) its complement is also an $O(k)$-clock \NTA language,
    and 2) any \DTA recognising it requires $2^k$ clocks.
    Thus, a deterministic separator may need exponentially many clocks in the size of the input \NTA.
\end{example}

In the rest of the section we show how to decide the separability problems.
We reduce the $k, m$-deterministic separability to $k, m$-timed synthesis,
and $k$-deterministic separability to $k$-timed synthesis,
for every fixed $k, m \in \N$.
Let $\A, \B$ be two \NTAe over alphabet $\Sigma$,
and let $\X$ be a set of $k$ clocks.
%
%
%
We build a timed synthesis game where the two sets of actions are
\begin{align*}
    A = \Sigma \quad \text{(Player \I),}
        \qquad
    B = \set{\acc, \rej} \quad \text{(Player \II)}.
\end{align*}
We define a projection function $\proj {a, b, t} = \tuple{a, t}$,
which is extended pointwise to finite and infinite timed words
$\proj {\tuple{a_0, b_0, t_0}\tuple{a_1, b_1, t_1} \cdots } = \tuple{a_0, t_0} \tuple{a_1, t_1} \cdots$
and timed languages $\proj L = \setof {\proj w} {w \in L \subseteq (A \cdot B \cdot \Rpos)^\omega}$.
Let $\Acc, \Rej \subseteq (A \cdot B \cdot \Rpos)^*$ be sets of those timed words
ending in a timed letter of the form $\tuple {\_, \acc, \_}$,
resp., $\tuple {\_, \rej, \_}$.
The winning condition for Player \I is
\begin{align}
    \label{eq:W0}
    W_0 = \big(\projinv {\lang \A} \cap \Rej \;\cup\; \projinv {\lang \B} \cap \Acc\big) \cdot (A \cdot B \cdot \Rpos)^\omega.
\end{align}
Crucially, we observe that $W_0$ is a $\NTAe$ language
since $\lang \A, \lang \B, \Rej, \Acc$ are so,
and this class is closed under inverse homomorphic images,
intersections, and unions.
The following lemma states the correctness of the reduction.
\begin{lemma}
    \label{lem:km:sep:red}
    There is a $k, m$-controller for Player \I in $\game A B{W_0}$ if, and only if,
    $\lang \A, \lang \B$ are $k, m$-deterministically separable.
\end{lemma}

\begin{proof}
    Let $\M = \tuple{A, B, \L, \ell_0, \delta}$ be a winning $k, m$-controller for Player \II in $G = \game{A}{B}{W_0}$.
    Let $\X = \set{\x_1, \dots, \x_k}$ be clocks of $\M$.
    We construct a separator 
    $\S = \tuple{\Sigma, \L \times B, \X, \I, \F, \Delta}$
    $\in$ \kmDTA{k}{m},
    where 
    $\I = \set{\tuple{\ell_0, \acc}}$ if $\varepsilon \in \lang \A$ and $\I = \set{\tuple{\ell_0, \rej}}$ otherwise,
    $\F = \L \times \set{\acc}$, and
    \begin{align}
        \label{eq:synthesis:sep}
        \transition{\tuple{\ell, b}} a \varphi \Y {\tuple{\ell', b'}} \in \Delta
            \quad \text{if, and only if,} \quad
                \delta(\ell, a, \varphi) = \tuple{\ell', b', \Y}.
    \end{align}
    We show that $\lang \S$ separates $\lang \A, \lang \B$
    using the fact that $\S$ is deterministic.
    In order to show $\lang \A \subseteq \lang \S$,
    let $w = \tuple{a_1, t_1} \cdots \tuple{a_n, t_n} \in \lang \A$
    and let Player \I play this timed word in $G$.
    Let the corresponding $\M$-conform partial play be
    $\pi = \tuple{a_1, b_1, t_1} \cdots \tuple{a_n, b_n, t_n}$.
    Since $\M$ is winning,
    $\pi$ does not extend to an infinite word in $W_0$,
    and in particular $\pi \not\in \projinv {\lang \A} \cap \Rej$.
    But $\proj \pi = w \in \lang \A$ by assumption, and thus $b_n = \acc$.
    The unique run of $\S$ on $w$ ends up in an accepting control location of the form $\tuple{\_, b_n}$,
    and thus $w \in \lang \S$, as required.
    The argument showing that $\lang \S \cap \lang \B = \emptyset$ is similar,
    using the fact that $\S$ is deterministic and must reach $b_n = \rej$
    and thus reject all words $\tuple{a_1, t_1} \cdots \tuple{a_n, t_n} \in \lang \B$.
    
    For the other direction,
    let $\S = \tuple{\Sigma, \L, \X, \set{\ell_0}, \F, \Delta} \in$ \kmDTA{k}{m} 
    be a deterministic separator.
    We construct a winning $k, m$-controller for Player \II in $G$
    of the form $\M = \tuple{A, B, \L, \ell_0, \delta}$
    where $\delta(\ell, a, \varphi) = \tuple{\ell', b, \Y}$
    for the unique $\Y, \ell', b$ \st
    $\transition \ell a \varphi \Y {\ell'} \in \Delta$
    and $b = \acc$ iff $\ell' \in \F$.
    In order to argue that $\M$ is winning in $G$,
    let $\pi = \tuple{a_1, b_1, t_1} \tuple{a_2, b_2, t_2} \cdots \in (A \cdot B \cdot \Rpos)^\omega$
    be an $\M$-conform play. By construction of $\M$ we have:
    \begin{claim}
    For every finite nonempty prefix $\pi' = \tuple{a_1, b_1, t_1} \cdots \tuple{a_n, b_n, t_n}$ of $\pi$,
    $\proj {\pi'} \in \lang \S$ if, and only if $b_n = \acc$.
    \end{claim}
    Knowing that $\lang \A \subseteq \S$, we deduce that no prefix of $\pi$ belongs to
    $\projinv {\lang \A} \cap \Rej$.
    Similarly, knowing that $\lang \S \cap \lang \B = \emptyset$, we deduce that no prefix of $\pi$ belongs to
    $\projinv {\lang \B} \cap \Acc$.
    Thus $\pi \notin W_0$ and therefore $\M$ is winning.
%
\end{proof}

\begin{proof}[Proof of \cref{sec:k:m:separability}]
    \Cref{lem:km:sep:red} provides a reduction from the $k,m$-deterministic separability problem
    to the $k,m$-timed synthesis problem.
    The latter problem is decidable by \cref{thm:km:synthesis}.
    Since the construction in \cref{lem:km:sep:red} is independent of $m$, it provides also 
    a reduction from the $k$-deterministic separability problem
    to the $k$-timed synthesis problem.
    The latter problem is decidable by \cref{thm:k:synthesis}.
\end{proof}

\section{Solving the timed synthesis problems}
\label{sec:timed:synthesis}

The second main result of this paper
is 
decidability of the $k,m$-timed synthesis problem 
and of the $k$-synthesis problem, i.e., when the maximal constant $m$ is not specified in advance
(\Cref{thm:km:synthesis,thm:k:synthesis}).
This will be achieved in four steps.
In the first two steps (see \cref{sec:zero:starting,sec:strictly:monotonic})
we make certain easy simplifying assumptions that winning conditions $W$ are
strictly monotonic, and \emph{zero-starting}: all words
$\tuple{a_1, t_1} \tuple{a_2, t_2} \cdots \in W$ satisfy $t_1 = 0$.
The main technical construction is in \cref{sec:km:synthesis},
where we prove \cref{thm:km:synthesis}
in such a way that we will easily obtain \cref{thm:k:synthesis}
as a corollary thereof in \cref{sec:k:synthesis}.

The decidability results of this section are tight,
since timed synthesis is undecidable when $k$ is not fixed (c.f.~\Cref{thm:synUnd}).

%


\subsection{Solving the $k, m$-timed synthesis problem}
\label{sec:km:synthesis}

In this section we prove \cref{thm:km:synthesis}
by reducing the $k, m$-timed synthesis problem
to a $0, 0$-timed synthesis problem,
which is decidable by \cref{lem:zero:synthesis}.
This is the most technically involved section.
The structure of the reduction will be useful in \cref{sec:k:synthesis}
to show decidability of the $k$-timed synthesis problem.
%

Let $\X$ be a fixed set of clocks of size $\card \X = k$
and let $m \in \N$ be a fixed bound on constants.
We reduce the $k, m$-synthesis problem to the $0,0$-synthesis problem by designing a protocol
in which Player \II, to compensate his inability to measure time elapse, can \emph{request} certain clocks to be \emph{tracked}.
In addition, we design the Player \I's winning condition that obliges her to remind whenever the value of any tracked clock is an integer,
by submitting \emph{expiry} information one time unit after a corresponding request.

Let $\fract{\x}$ stand for the fractional part of the value of a clock $\x$.
For $\Y_1, \Y_2 \subseteq \X$, two (partial) clock valuations $\mu \in \Rpos^{\Y_1}, \nu\in\Rpos^{\Y_2}$ 
are \emph{fractional region equivalent} if $\Y_1 = \Y_2$ and they exhibit the same relations
between fractional parts of clocks: $\mu \models \fract {\x} < \fract{\x'}$ iff
$\nu \models \fract {\x} < \fract{\x'}$  and
$\mu \models \fract {\x} = 0$ iff $\nu \models \fract {\x} = 0$, for all $\x, \x' \in \Y_1$.
By a (partial) fractional $\X$-region $\f$ we mean an equivalence class of this equivalence relation. 
%
%
All elements $\mu \in \Rpos^\Y$ in $\f$ have the same domain $\Y$,
which we denote by $\dom \f = \Y$.
Let $\one \f = \setof{\x\in\dom \f}{\f \models \fract \x = 0}$.
Let $\fregions \X$ be the set of all fractional $\X$-regions, including the empty one $\f_0$ with $\dom {\f_0} = \emptyset$.
For $\r \in \regions \X m$ and $\f \in \fregions \X$,
we say that $\f$ \emph{agrees} with $\r$ if they give the same answer for clocks 
$\x, \y \in \dom \f$:%
\begin{itemize}
    \item $\f \models \fract \x < \fract \y$ if, and only if,
    $\r \models \fract \x < \fract \y \vee \x > m \vee \y > m$;
    \item $\f \models \fract \x = 0$ if, and only if, $\r \models \fract \x = 0 \vee \x > m$.
\end{itemize}
%
\noindent
The successor relation between regions induces a corresponding relation between fractional regions: $\f \preceq \f'$
whenever $\dom \f = \dom {\f'}$, $\f$ agrees with some $\r$, $\f'$ agrees with some $\r'$, and $\r \preceq \r'$.
The immediate successor is the minimal $\f'$ with $\f \prec \f'$.
Finally, the \emph{successor region of $\r$ agreeing with $\f$} is
%
	$\xsuccessor {\X, m} \r {\f} = \min_\preceq \setof{\r' \succeq \r}{\f \text { agrees with } \r'}$.
%
In the sequel we apply clock resets also to regions 
$\extend \Y 0 \r$ and fractional regions.

Let the original game $G = \game A B W$ have action alphabets $A, B$
and Player \I's winning condition $W \subseteq (A \cdot B \cdot \Rpos)^\omega$.
Thanks to \cref{sec:strictly:monotonic,sec:zero:starting}
we assume that $W$ is both strictly monotonic and zero starting.
We design a new game $G' = \game {A'} {B'} {W'_{k, m}}$ as follows.
We take as the new action alphabets the sets
\begin{align}
    \label{eq:actions}
    A' \;=\; (A \cup \set{\tick}) \cdot \fregions \X
        \quad \text{and} \quad 
            B' \;=\; (B \cup \set{\tick}) \cdot 2^\X.
\end{align}
The players' action sets $A', B'$ depend only on the set of clocks $\X$
and do not depend on the maximal constant $m$.
Moves of the form $\tuple{\tick, \_}$ are \emph{improper}
and the other ones (i.e., those involving an $A$ or $B$ component) are \emph{proper}.
%
%
Let an infinite play be of the form
\begin{align}
    \label{eq:play}
    \pi =
        \tuple{a_1', b_1', t_1} \tuple{a_2', b_2', t_2} \cdots \in (A' \cdot B' \cdot \Rpos)^\omega,
    \text{ with }
        a_i' = \tuple{a_i, \f_i}
    \text{ and }
        b_i' = \tuple{b_i, \Y_i}.
\end{align}
%
The domain $\T_i = \dom {\f_i}$ of a fractional region 
denotes the clocks \emph{tracked} at time $t_i$,
i.e., those for which Player \I needs to provide expiry information.
%
%
Sets $\Y_i$'s denote clocks which Player \II wants to be continued to be tracked:
by an \emph{$\x$-request at time $t_i$} we mean a Player \II's move $b_i'$ with $\x \in \Y_i$.
An $\x$-request at time $t_i$ is \emph{cancelled} if 
there is another $\x$-request for the same clock at some time $t_i < u < t_i + 1$.
%
An \emph{improper $\x$-request chain} starting at time $t_i$ of length $l\geq 1$
is a sequence of improper non-cancelled $\x$-requests at times
$t_i$, $t_i + 1$, \dots, $t_i + l - 2$,
followed by an improper (but possibly cancelled) $\x$-request at time $t_i+l-1$.
%
%
Likewise one defines an infinite improper $\x$-request chain starting at time $t_i$.


\newcommand{\fontsizeimg}{\normalsize}
\definecolor{accentGray}{RGB}{71, 71, 75}
\definecolor{accent}{RGB}{252, 199, 18}
\tikzstyle{letter} = [inner sep=0.8mm, anchor=north west, yshift=-2mm]
\tikzstyle{bgLine} = [draw=accentGray, line width=0.2mm]
\tikzstyle{box} = [bgLine, fill=white, inner sep=0.8mm, anchor=north west, outer sep=0, text=black]
\tikzstyle{dot} = [fill=black, inner sep=0.8mm, outer sep=0]
\tikzstyle{axis} = [line width=0.5mm, draw=accent, -{Triangle[scale=1]}]

\begin{example}
    Before defining the winning set $W'_{k,m}$ formally, we illustrate the underlying idea.
    Consider the following partial play
    $(a_1, b_1, 0)\,(a_2, b_2, 4.2)\,(a_3, b_3, 6) \in (A \cdot B \cdot \Rpos)^*$
    in $G$:
    %
    %
    \begin{center}\begin{tikzpicture}[xscale=1.95]
        \node[text width=1.5cm, anchor=east, align=right] at (6.9,0.03) {\footnotesize \textsf{\phantom{emulated}\\time}};
        \foreach \x in {0,...,6} {
            \draw[bgLine] (\x,-0.2) -- (\x,0.2);
            \node at (\x, 0.5) {\fontsizeimg{$\x$}};
        }
        \draw[axis] (0,0) -- (7,0);

        \foreach \letterA/\time/\letterB in {a_1/0/b_1,a_2/4.2/b_2,a_3/6/b_3} {
            \node[dot] (dot) at (\time,0) {};
            \node[box] (box) at (\time,-0.4) {
    \footnotesize{\!\!$\begin{array}{c}\letterA\\\letterB\end{array}$\!\!}
            };
            \draw[bgLine] (\time, -0.08) -- (\time, -0.4);
        }
    \end{tikzpicture}
    \end{center}
    In $G'$, 
    Player \II demands Player \I to provide clock expiry information.
    %
    Let $\X = \{\x, \mathtt{y}\}$ and $m=3$.
    Suppose Player \II wants to make sure that $a_2$ comes at time $> 3$.
    To this end, she makes an $\x$-request chain of length 3
    (we write $\overline{\x}$ instead of $\fract{\x}$;
    $\f_{\phi}$ denotes the fractional $\X$-region agreeing with $\phi$):
\vspace{-5mm}
    \begin{center}
    \begin{tikzpicture}[xscale=1.95]
        \node[text width=1.5cm, anchor=east, align=right] at (6.9,0.03) {\footnotesize \textsf{emulated\\time}};
        \foreach \x in {0,...,6} {
            \draw[bgLine] (\x,-0.2) -- (\x,0.2);
            \node at (\x, 0.5) {\fontsizeimg{$\x$}};
        }
        \draw[axis,dashed] (0,0) -- (7,0);
        \foreach \letterA/\time/\letterB/\name in {%
            (a_1, \f_0)       /0  /{(b_1, \{\x\})}/x,%
            (a_2, \f_0)       /4.2/{(b_2, \{\y\})}/y,%
            (a_3, \f_{0<\overline{\y}})/6  /{(b_3, \{\y\})}/z%
        } {
            \node[dot] at (\time,0) {};
            \node[box] (request\name) at (\time,-0.4) {
                \footnotesize{\!\!\!$\begin{array}{l}\letterA\\\letterB\end{array}$\!\!\!}
            };
            \draw[bgLine] (\time, -0.08) -- (\time, -0.4);
        }
        \foreach \time/\knowledge in {0/\x=0, 1/\x=1, 2/\x=2, 3/\x=3, 4.2/{\x>3, \y=0}, 5.2/\y=1, 6/1<\y<2} {
            \node[anchor=north west, bgLine, draw, dotted, text=accentGray] at (\time-0.005, -3.1) {\footnotesize{$\strut{{\knowledge}}$}};
            \draw[bgLine,dotted] (\time,-3.2) -- (\time,-1);
        }
        \foreach \letterA/\letterB/\time/\name in {%
            {(\tick, \f_{\overline{\x}=0})}/{(\tick, \set{\x} )}/1  /x1,%
            {(\tick, \f_{\overline{\x}=0})}/{(\tick, \set{\x} )}/2  /x2,%
            {(\tick, \f_{\overline{\x}=0})}/{(\tick, \emptyset)}/3  /x3,%
            {(\tick, \f_{\overline{\y}=0})}/{(\tick, \set{\y} )}/5.2/y1%
        } {
            \node[box] (request\name) at (\time,-1.35) {
                \footnotesize{\!\!\!$\begin{array}{l}\letterA\\\letterB\end{array}$\!\!\!}
            };
            \node[dot,fill=white,draw] at (\time,0) {};
            \draw[bgLine] (\time, -0.08) -- (\time, -1.35);
        }
        \foreach \time/\timee in {x/x1,x1/x2,x2/x3,y/y1} {
            \draw[-{Triangle},thick,rounded corners=1mm] ($(request\time.south)!0.2!(request\time.south east)$) |- ($(request\timee.south west)!0.8!(request\timee.south)-(0,0.3)$) -- ($(request\timee.south west)!0.8!(request\timee.south)$);
        }
        \draw[thick,rounded corners=1mm] ($(requesty1.south)!0.2!(requesty1.south east)$) |- (6,-2.477);
        \draw[thick,rounded corners=1mm] ($(requestz.south)!0.2!(requestz.south east)$) |- (7,-2.077);

        \draw (0,-2.7) -- (7,-2.7);
        \draw[thick,draw=red,{Triangle}-] (6,-2.477) -- (requestz.south west);
        \node[outer sep=0.0mm, rotate=0, anchor=south west, text width=3cm] at (6,-2.7) {\footnotesize\textsf{\color{red}{cancelled}\phantom{y}}};
        \node[outer sep=0.0mm, rotate=0, anchor=south, text width=3cm] at (4.3,-2.7) {\footnotesize\textsf{what is played}};
        \node[outer sep=0.0mm, rotate=0, anchor=north, text width=3cm] at (4.3, -2.7) {\footnotesize\textsf{Player \II's knowledge}};
    \end{tikzpicture}
    \end{center}
    %
    %
    The length of an $\x$-chain at any given moment
    corresponds to the integral part of $\x$;
    the expiry information for $\x$ is provided by Player \I
    precisely when the fractional part of $\x$ is $0$.
\end{example}

In order to define $W'_{k,m}$
it will be convenient to have the following additional data extracted from $\pi$.
%
%
Let $\delta_i = t_i - t_{i-1} \geq 0$ be the time elapsed by Player \I at round $i$ (with $t_0 = 0$).
%
%
Furthermore, let $\nu_0 = \lambda\x \cdot 0$ be the initial clock valuation,
and, for $i \geq 0$, let 
%
\begin{align} \label{eq:nudef}
\nu_{i+1} = \extend {\Y_{i+1}} 0 {(\nu_i + \delta_{i+1})}.
\end{align}
In words, every $\x$-request is interpreted as reset of clock $\x$.
The winning condition $W'_{k,m}$ in the new game will impose, in addition to $W$,
the following further conditions to be satisfied by Player \I in order to win.
Let $W^\I_k \subseteq (A' \cdot B' \cdot \Rpos)^\omega$
be the set of plays $\pi$ as in \eqref{eq:play} which are zero-starting ($t_1 = 0$), strictly monotonic and, for every $i \geq 1$:
\begin{enumerate}

    \item For every $\x \in \X$, $\x$ is expired at time $t_i$ if, and only if, $t_i \geq 1$ and
    there is a non-cancelled $\x$-request at an earlier time $t_j = t_i-1$.
    \label{cond:x}

    \item Tracked clocks are consistent with requests:
    for every clock $\x\in \X$, 
    $\x$ is tracked $\x \in \T_i$ at time $t_i$ if, and only if, 
    there is an $\x$-request at an earlier $t_j$ with
    $t_i - 1 \leq t_j < t_i$.
    \label{cond:ZZ}

    \item The fractional regions are correct: $\f_i$ agrees with
    $\region \X m {(\nu_{i-1} + \delta_{i})}$.
    %
    \label{cond:leq}

    \newcounter{enumTemp}
    \setcounter{enumTemp}{\theenumi}
\end{enumerate}

\noindent
Thus the conditions above assure that Player \I provides exactly all expiry information
requested by Player \II in a timely manner,
and the fractional regions $\f_i$ are consistent with the requests and time elapse.
Note that any play in $W^\I_k$ satisfies $0 < \nu_i(\x) \leq 1$ for every $i\geq 1$ and $\x\in\T_i$. Indeed,
positivity is due to strict monotonicity, and the upper bound due to the conditions \ref{cond:x}--\ref{cond:leq}.
Provided Player \I satisfies $W^\I_k$,
she wins whenever Player \II violates any of the conditions below:
Let $W^\II_{k,m} \subseteq (A' \cdot B' \cdot \Rpos)^\omega$
be the set of plays $\pi$ as in \eqref{eq:play} s.t.
\begin{enumerate}
    \setcounter{enumi}{\theenumTemp}
    
    \item Player \II plays a proper move iff Player \I does so.
    \label{cond:proper}
    
    \item Every improper Player \II's $\x$-request $b_i'$
    is a response to Player \I's expiry information for $\x$:
    $\Y_i \subseteq \one {\f_i}$.
    (Proper $\x$-requests are allowed unconditionally.)
    \label{cond:Y}

    \item For every clock $\x \in \X$,
    the length Player \II's improper $\x$-request chains is $< m$. 
    This is the only component in the winning condition which depends on $m$.
    \label{cond:chain}

\end{enumerate}

\noindent
Consider the projection function
$\phi : (A' \cdot B' \cdot \Rpos) \to (A \cdot B \cdot \Rpos) \cup \set \varepsilon$
\st $\phi(\tuple{a, \_}, \tuple{b, \_}, t) = \varepsilon$
if $a = \tick$ or $b = \tick$,
and $\phi(\tuple{a, \_}, \tuple{b, \_}, t) = \tuple{a, b, t}$
if $a \in A$ and $b \in B$,
which is extended homomorphically on finite and infinite plays.
The winning condition for Player \I in $G'$ is
\begin{align}
    \label{eq:winning:condition}
    W'_{k, m} \;=\;
        W^\I_k \cap \left(
            \phi^{-1}(W) \cup
            (A' \cdot B' \cdot \Rpos)^\omega \setminus W^\II_{k, m}\right).
\end{align}
Since $W$, $W^\I_k$, are \NTAe languages,
and $W^\I_k$ and $W^\II_{k,m}$ are \kDTA{k} languages over $A' \cdot B'$, 
thanks to the closure properties \DTA and \NTAe languages
the winning condition $W'_{k, m}$ is an \NTAe language.
%
%
In what follows, an \emph{untimed controller} is a $0,0$-controller.
Then next two lemmas state the correctness of the reduction.
Our assumption on strict monotonicity facilitates the correctness proof since we
need not deal with simultaneous events.
\begin{lemma} \label{lem:km:untimed}
    If there is a winning $k, m$-controller $\M$ for $G$, then
    there is a winning untimed controller $\M'$ for $G'$.
\end{lemma}
\begin{proof}
    Let $\M = \tuple{A, B, \L, \ell_0, \delta}$ be a winning $k, m$-controller $\M$ for $G$ with clocks $\X = \set{\x_1, \dots, \x_k}$ and
    update function $\delta : \L \cdot A \cdot \regions \X m \to \L \cdot B \cdot 2^\X$.
    %
    %
    We define a winning untimed controller $\M' = \tuple{A', B', \L', \rhd, \delta'}$ for $G'$
    with memory locations 
    %
        $\L' = \set{\rhd} \ \cup \ \L \cdot \regions \X m$,
    %
    where $\rhd$ is the initial memory location, and remaing memory locations 
    are of the form $(\ell, \r)$, where
    $\ell \in \L$ is the current memory location of $\M$
    and $\r \in \regions \X m$ is the current region of $\M$'s clocks.
    The update function $\delta' : \L' \cdot A' \to \L' \cdot B'$
    (we omit regions and clock resets because $\M'$ has no clocks)
    is defined as follows.
    As long as the play is in $W^\I_k$, we can assume that Player \I 
    starts with $((a, \f_0), t)$ and $t=0$,
    due to the zero-starting restriction,
    which allows Player \II to submit requests at time $0$. 
    Consequently, let
    $ \delta'(\rhd, \tuple{a, \f}) =
            \tuple{\tuple{\ell', \zeroregion}, \tuple{b, \X}}$,
    where the next location $\ell'$ and the response $b$ are determined by $\delta(\ell_0, a, \zeroregion) = (\ell', b)$, and
    the set $\X$ denotes a request to track all clocks.
    Then, for every $\ell, \r, a, \f$, let
    \begin{align} \label{eq:rhs}
        \delta'(\tuple{\ell, \r}, \tuple{a, \f}) =
            \tuple{\tuple{\ell', \r'}, \tuple{b, \Y}},
    \end{align}
    where the r.h.s.~is defined as follows.
    %
    Let $\T = \dom \f$ be the currently tracked clocks, and $\T_0 = \one \f \subseteq \T$ the currently expired ones.
    If $\f$ agrees with no successor region of $\r$ 
    then Player \II wins immediately because Player \I is violating condition
    \ref{cond:leq}.
    Therefore, assume such a successor region $\hat \r = \xsuccessor {\X, m} \r {\f}$ exists.
    We do a case analysis based on whether Player \I plays a proper or an improper move.
    \begin{itemize}

        \item Case $a \in A$ (proper move):
        Let $\delta(\ell, a, \hat \r) = (\ell', b, \Y)$ thus defining $\ell'$ and $(b, \Y)$ in~\eqref{eq:rhs}.
        %
        %
        Take as the new region $\r' = \extend \Y 0 {\hat \r}$.

        \item Case $a = \tick$ (improper move):
        Let the response be also improper $b = \tick$,
        the control location does not change $\ell' = \ell$,
        the new clocks to be tracked are the expired clocks with a short improper chain
        $\Y = \setof{\x \in \T_0}{\hat\r \models \x = 1 \vee \cdots \vee \x = m - 1}$,
        and $\r' = \hat \r$. 

    \end{itemize}
    
    \noindent
    Consider an infinite $\M'$-conform run in $G'$
    (omitting clock valuations since $\M'$ has no clocks)
    \begin{align*}
        \rho' \; = \;
            \rhd
            \tuple{a'_1, b'_1, t_1, \tuple{\ell_1, \r_1}}
            \tuple{a'_2, b'_2, t_2, \tuple{\ell_2, \r_2}} \cdots \in \iRun {\M'},
            a'_i = \tuple{a_i, \f_i},
            b'_i = \tuple{b_i, \Y_i}.
    \end{align*}
    If the induced play $\pi' = \rtop {\rho'} =
            \tuple{a'_1, b'_1, t_1}
            \tuple{a'_2, b'_2, t_2} \cdots \in (A' \cdot B' \cdot \Rpos)^\omega$ 
    is not in $W^\I_k$,
    then Player \II wins and we are done.
    Assume $\pi' \in W^\I_k$, and thus conditions \ref{cond:x}--\ref{cond:leq} are satisfied.
    We argue that $\pi' \in W^\II_{k, m}$.
    The conditions~\ref{cond:proper} and~\ref{cond:Y} hold by construction.
    Aiming at demonstrating that~\ref{cond:chain} holds too,
    let $\mu_0 = \lambda\x \cdot 0$,
    and, for $i \geq 0$, let 
    \begin{align} \label{eq:mudef}
    	\mu_{i+1} = \begin{cases}
    				{\mu_i + \delta_{i+1}} & a_i = \tick \quad \text{ (improper round)} \\
    				\extend {\Y_i} 0 {(\mu_i + \delta_{i+1})} & a_i \in A \quad \text{ (proper round)}.
			   \end{cases}
    \end{align}
    Thus clock valuations $\mu_i$ are defined exactly as $\nu_i$ in~\eqref{eq:nudef} except that only proper requests are interpreted as clock resets. 
    We claim that the region information $\r_i$ 
    is consistent with $\mu_i$:
    %
         $\r_i = \region \X m {\mu_i}$ (*).
 \noindent
 Indeed, this is due to $\pi' \in W^\I_k$, and the fact that $\M'$ updates its stored region consistently with
 time elapse: at every round $\M'$ uses the successor region agreeing with the current fractional region submitted by Player \I, and
 resets a set of clocks $\Y$ exactly
 when she plays a proper move of the form $(a, \Y)\in A \cdot 2^\X$.
 Since an $\x$-request is submitted by $\M'$ only when $\hat r \models x\leq m{-}1$,
 condition~\ref{cond:chain} holds.
 
    In order to show that Player \II is winning, consider an $\M'$-conform run $\rho'$. 
    It suffices to show $\pi' = \rtop {\rho'} \not\in \phi^{-1}(W)$.
    %
    %
    Let the proper moves in $\rho'$ be at indices $1 = i_1 < i_2 < \cdots$
    ($i_1 = 1$ due to zero-starting).
    In particular, $\ell_{i_l} = \ell_{i}$ for $i_l \leq i < i_{l+1}$.
    Consider the run
    $    \rho \; = \;
            \tuple{\ell_{0}, \mu_0}
            \tuple{a_{i_1}, b_{i_1}, t_{i_1}, \tuple{\ell_{i_1}, \mu_{i_1}}}
            \tuple{a_{i_2}, b_{i_2}, t_{i_2}, \tuple{\ell_{i_2}, \mu_{i_2}}} \cdots$.
    Using (*) 
    and the definition of $\M'$,
    one can prove by induction that $\rho$ is an $\M$-conform run in $G$.
    Since $\M$ is winning, the induced play
        $\pi = 
            \rtop \rho =
            \tuple{a_{i_1}, b_{i_1}, t_{i_1}}
            \tuple{a_{i_2}, b_{i_2}, t_{i_2}} \cdots \in (A \cdot B \cdot \Rpos)^\omega$,
    satisfies $\pi \notin W$.
    Again by induction one can prove that $\pi = \phi(\pi')$.
    Hence $\phi(\pi') \notin W$ as required. 
\end{proof}

\begin{restatable}{lemma}{lemSecondImpl}
	\label{lem:Second:Impl}
    If there is a winning untimed controller $\M'$ in $G'$,
    then     there is a winning $k, m$-controller $\M$ in $G$.
\end{restatable}
%


\subsection{Solving the $k$-timed synthesis problem}
\label{sec:k:synthesis}

In this section we prove \cref{thm:k:synthesis},
stating that the $k$-timed synthesis problem is decidable,
by reducing it to the $0, 0$-synthesis problem,
which is decidable by \cref{lem:km:untimed}.
%
We build on the game defined in \cref{sec:km:synthesis}.
Starting from a timed game $G = \game A B W$
we define the timed game $G'' = \game {A'} {B'} {W''_k}$,
where the sets of actions $A'$ and $B'$ are as in \eqref{eq:actions},
and the winning condition $W''_k$ is defined as follows.
Let $W^\II_k \subseteq (A' \cdot B' \cdot \Rpos)^\omega$
be the set of plays where, for every clock $\x \in \X$,
improper $\x$-request chains have finite lengths:
$W^\II_k = \bigcup_{m \in \N} W^\II_{k, m}$.
(In other words, $(A' \cdot B' \cdot \Rpos)^\omega \setminus W^\II_k$
contains plays with an infinite improper $\x$-request chain, for some clock $\x \in \X$.)
Then, $W''_k$ is defined as $W'_{k,m}$ from \eqref{eq:winning:condition},
except that $W^\II_{k,m}$ is replaced by the weaker condition $W^\II_k$
(notice $W''_k$ does not depend on $m$):
\begin{align}
    \label{eq:winning:condition:two}
    W''_k \;=\;
        W^\I_k \cap \left(
            \phi^{-1}(W) \cup
            (A' \cdot B' \cdot \Rpos)^\omega \setminus W^\II_k\right).
\end{align}
\vspace{-7mm}
\begin{restatable}{lemma}{lemmaUntimedUntimed} \label{lemma:untimed:untimed}
    There is a winning untimed controller for $G''$ if, and only if,
    there is some $m\in\N$ and a winning untimed controller for $G' = \game {A'} {B'} {W'_{k,m}}$.
\end{restatable}

\begin{proof}
    For the ``if'' direction,
    we observe that $W''_k \subseteq W'_{k,m}$, for every $m\in\N$.
    Hence every winning untimed controller for $G'$ is also winning for $G''$.
    For the ``only if'' direction,
    let $\M'' = \tuple{A', B', \L, \ell_0, \delta}$ be an untimed winning controller in $G''$.
        Let $m = \card {A'} \cdot \card \L + 1$.
        We claim that $\M''$ is also winning in $G' = \game {A'} {B'} {W'_{k,m}}$ for this choice of $m$.
    Towards reaching a contradiction,
    suppose $\M''$ is losing in $G'$.
    An $\M''$-conform run $\rho$ in $G'$ (or in $G''$) and its associated play $\pi$ are of the form
    \begin{align*}
        \rho \;=\; &\ell_0 \tuple{a'_1, b'_1, t_1, \ell_1}
            \tuple{a'_2, b'_2, t_2, \ell_2} \cdots \in \iRun{\M''},
            \text{ with } 
            a'_i = \tuple{a_i, \f_i} \text { and } b'_i = \tuple{b_i, \Y_i},
             \\
        \pi \;=\; &\rtop \rho = \tuple{a'_1, b'_1, t_1} \tuple{a'_2, b'_2, t_2} \cdots \in \Play{\M''}.
    \end{align*}
    Let $\rho_i \in \Run{\M''}$ be the finite prefix of $\rho$
    ending at $\tuple{a'_i, b'_i, t_i, \ell_i}$.
    %
    %
    Since $\M''$ is losing in $G'$, some $\M''$-conform play $\pi$ above is in $W'_{k, m}$.
    %
    Since $\M''$ is winning in $G''$, $\pi \not\in \phi^{-1}(W)$,
    and thus $\pi \in W^\I_k \setminus W^\II_{k, m}$.
    This means that $\pi$ contains an improper $\x$-request chain $C$ of length $m$,
    for some clock $\x \in \X$.
    %
    %
    By the definition of $m$, there are indices $i < j$
    \st the the same controller memory repeats together with Player \I's action
    $\tuple{a'_i, \ell_i} = \tuple{a'_j, \ell_j}$. In particular $\f_i = \f_j$.
    Since $\M''$ is deterministic and its action depends only on Player \I's action $a'_i$
    and control location $\ell_i$,
    a posteriori we have $b'_i = b'_j$ as well.
    Moreover, as consecutive timestamps in $C$ are equal to the first one plus consecutive nonnegative integers, 
    $\Delta = t_i - t_j \in \set{1, \dots, m - 1}$.
    Consider the corresponding infix
    $\sigma = \tuple{a'_{i+1}, b'_{i+1}, t_{i+1}, \ell_{i+1}} \cdots \tuple{a'_j, b'_j, t_j, \ell_j}$
    of the run $\rho$.
    Since $\pi \in W'_{k,m}$,
    thanks to conditions \ref{cond:ZZ} and \ref{cond:leq} the fractional regions $\f_i = \f_j$ contain all tracked clocks,
    and they agree with the clock valuations $\nu_i$ and $\nu_j$, respectively, as defined in~\eqref{eq:nudef}.
    Let
    $\set{t_i-1 \leq t_{i_1} < t_{i_2} < \cdots < t_{i_l} < t_i} = \setof{t_i - \nu_i(\x)}{\x \in \dom {\f_i}}$
    be the timestamps corresponding to the last request of the clocks tracked at time $t_i$,
    and likewise let
    $\set{t_j-1 \leq t_{j_1} < t_{j_2} < \cdots < t_{j_{l'}} < t_j} = \setof{t_j - \nu_j(\x)}{\x \in \dom {\f_j}}$. 
    By assumption, $\f_i = \f_j$, and hence $l = l'$ and for $\x\in\dom{\f_i} = \dom{\f_j}$ and $1\leq h \leq l$, 
    %
        $t_{i_h} = t_i - \nu_i(\x) \text{ if, and only if, } t_{j_h} = t_j - \nu_j(\x)$ (*).
    %
    Moreover, since $\one{\f_i} = \one{\f_j}$, we have
    %
        $t_{i_1} = t_i -1 \text{ if, and only if, } t_{j_1} = t_j - 1$ (**).
    %
    Player I will win in $G'$
    by forcing a repetition of the infix $\sigma$ \emph{ad libitum}.
    In order to do so, we need to modify its timestamps.
    An \emph{automorphism} of the structure $\tuple{\R, \leq, +1}$
    is a monotonic bijection preserving integer differences,
    in the sense that $f(x+1) = f(x) + 1$ for every $x\in\R$.
    Note that such an automorphism is uniquely defined by its action on any unit-length interval.
    We claim that there exists such an automorphism $f: \R \to \R$
    mapping $t_i-1$ to $t_j-1$ (and hence forcedly also $t_i$ to $t_j$), and each $t_{i_h}$ with $1 \leq h \leq l$ to $f(t_{i_h}) = t_{j_h}$.
    %
    %
    This is indeed the case,
    by (*) and (**)
    all timestamps $t_{i_h}$'s belong to the unit half-open interval $[t_i - 1, t_i)$
    and likewise all timestamps $t_{j_h}$'s belong to $[t_j -1, t_j)$.
    We apply $f$ to a timed word $\sigma \mapsto f(\sigma)$
    by acting pointwise on timestamps.
    Consider the infinite run
    %
    $    \rho' \;=\;
            \rho_i \cdot
                \sigma \cdot
                    f(\sigma) \cdot
                        f(f(\sigma)) \cdots;$
    %
    it is $\M''$-conform since the controller $\M''$ is deterministic.
    %
    By construction, $\rho'$ contains an infinite $\x$-request chain,
    and thus $\rho' \not\in W^\II_k$.
    It remains to argue that $\rho \in W^\I_k$ implies $\rho' \in W^\I_k$ as well.
    Let there be a non-cancelled $\x$-request at time $t_s$ in $\rho'$.
    If $t_s < t_j - 1$,
    then this request must be satisfied at time $t_{s'} = t_s + 1 < t_j$,
    and thus already in $\rho_i \cdot \sigma$,
    which is the case since the latter is a prefix of $\rho \in W^\I_k$.
    Now assume $t_j - 1 \leq t_s < t_j$.
    Thus $t_s = t_{j_h}$ for some $1 \leq h \leq l$.
    By the definition of $f$,
    $f^{-1}(t_s) = t_{i_h} < t_j - 1$ and, thanks to the previous case,
    the request at $t_{i_h}$ is satisfied at $t_{i_h} + 1$ due to (*).
    By applying $f$ we obtain $f(t_{i_h} + 1) = f(t_{i_h}) + 1 = t_s + 1$,
    and thus the request at time $t_s$ is satisfied at time $t_s + 1$ in $f(\sigma)$,
    as required.
    The general argument for 
    $t_j + n \Delta + d - 1 \leq t_s < t_j + n \Delta + d$, where $n\geq 0$ and $0\leq d < \Delta$, is similar,
    using induction on $n$.
\end{proof}

\begin{proof}[Proof of \cref{thm:k:synthesis}]
    Due to \cref{lem:km:untimed,lem:Second:Impl,lemma:untimed:untimed},
    there is a winning untimed controller $\M''$ for $G''$ if, and only if 
    there is some $m\in\N$ and a winning $k,m$-controller $\M$ for $G$.
    Thus the $k$-synthesis problem reduces to the $0,0$-synthesis problem,
    and the latter is decidable thanks to \cref{lem:zero:synthesis}.
\end{proof}


\section{Future work}

While deterministic separators may need exponentially many clocks
(c.f.~\cref{example:separable:exponential}),
we do not have a computable upper bound on the number of clocks of the separating automaton (if one exists).
We leave the \DTA separability problem
when the number of clocks is not fixed in advance as a challenging open problem.
In this case,
we cannot reduce the separability problem to a timed synthesis problem,
since the latter is undecidable.
\begin{restatable}{theorem}{thmSynUnd}
    \label{thm:synUnd}
    The timed synthesis problem is undecidable,
    and this holds already when Player \I's winning condition is a \kNTA 1 language.
\end{restatable}
\noindent
We leave the computational complexity of separability as future work.

Deterministic separability can be considered also over infinite timed words.
We chose to present the case of finite words
because it allows us to focus on the essential ingredients of this problem.
When going to infinite words, new phenomena appear already in the untimed setting;
for instance, deterministic B\"uchi automata are less expressive than deterministic parity automata,
and thus one should additionally specify in the input which priorities can be used by the separator;
or leave them unspecified and solve a more difficult problem.

Analogous results about separability of register automata can be obtained with techniques similar to the one presented in this paper.
We leave such developments for further work.

\bibliography{bib}

\appendix


\section{Missing proofs in Section~\ref{sec:synthesis}}

We first define synthesis games in the untimed setting,
and then formally show that the timed synthesis problem for $0, 0$-controllers is decidable
by reduction to the untimed setting.

\para{Synthesis games}
Let $A$ and $B$ be two finite alphabets of actions
and let $W \subseteq (A \cdot B)^\omega$
be a language of $\omega$-words over the alphabet $A \cdot B$.
The \emph{synthesis game} is played by Player \I and Player \II in rounds.
At round $i \geq 0$,
Player \I chooses an action $a_i \in A$
and then Player \II chooses a \emph{response} $b_i \in B$.
The game is played for $\omega$ rounds,
and at doomsday the two players have produced an infinite play
%
   $ \pi = a_1 b_1 a_2 b_2 \cdots \in (A \cdot B)^\omega.$
%
Player \I wins the game if, and only if, $\pi \in W$.

A \emph{controller} for Player \II
is a Mealy machine of the form $\M = \tuple{A, B, \L, \ell_0, \delta}$
where $\L$ is a finite set of memory locations,
$\ell_0 \in \L$ is the initial memory location,
and $\delta : \L \cdot A \to \L \cdot B$ is the update function
mapping the current memory $\ell \in \L$ and input $a \in A$,
to $\delta(\ell, a) = \tuple{\ell', b}$,
where $\ell' \in \L$ is the next memory location and $b \in B$ is an output symbol.
We define by mutual induction the notion of \emph{$\M$-conform partial runs}
$\Run \M  \subseteq \L \cdot (A \cdot B \cdot \L)^*$
and the strategy $\sem \M : \Run \M \cdot A \to \L \cdot B$ induced by the controller on conform runs as follows:
Initially, $\ell_0 \in \Run \M$.
Inductively, for every $n > 0$ and every $\M$-conform partial run
$\pi = \ell_0 (a_1 b_1 \ell_1) \cdots (a_n b_n \ell_n) \in \Run \M$,
for every $a \in A$,
$\sem \M (\pi \cdot a) = \tuple{\ell', b}$
for the unique $\ell', b$ \st 
$\delta(\ell_n, a) = \tuple{\ell', b}$,
and $\pi \cdot (a  b  \ell') \in \Run \M$.
An infinite $\M$-conform run is any sequence $\pi \in  \L \cdot (A \cdot B \cdot \L)^\omega$
such that every finite prefix thereof is $\M$-conform. By $\rtop \pi$ we denote the infinite play obtained from $\pi$
by dropping locations.

The \emph{synthesis problem} amounts to decide, given $A, B$ and an $\omega$-regular language
$W\subseteq (A\cdot B)^\omega$, whether there is a controller $\M$
\st every infinite $\M$-conform run $\rho$
satisfies $\rtop \rho \not\in W$.

\begin{theorem}[\protect{\cite[Theorem $1'$]{BuchiLandweber:AMS:1969}}]
    \label{thm:BuchiLandweber}
    The synthesis problem is decidable.
\end{theorem}

\lemZeroSynthesis*
\begin{proof}
    Consider a timed synthesis game $\game A B W$ and let $W' = \untime W$.
    Winning $0, 0$-controllers in $\game A B W$
    are in one-to-one correspondence with winning controllers
    in the corresponding untimed synthesis game with winning condition $W'$.
    Indeed, the update function $\delta : \L \cdot A \cdot \regions k m \to \L \cdot B \cdot 2^\X$
    of a $k, m$-controller $\M$ when $k = m = 0$
    can equivalently be presented as a function of type $\L \cdot A \to \L \cdot B$
    (which we take as the update function in the untimed controller $\M'$),
    and all functions of the latter type arise in this way.
    If $\M$ is losing in $\game A B W$,
    then there is a $\M$-conform run $\rho \in W$,
    and thus $\untime \rho$ is a $\M'$-conform run in $W'$,
    showing that $\M'$ is losing in the corresponding untimed synthesis game.
    On the other hand, let $\rho' \in W'$ be $\M'$-conform.
    Since $\M$ does not look at the timestamps,
    we can choose them accordingly in order to find an $\M$-conform timing thereof
    $\rho \in \untimeinv {\rho'} \cap W$.
    Untimed synthesis is decidable by \cref{thm:BuchiLandweber}.
\end{proof}

\section{Missing proofs in Section~\ref{sec:timed:synthesis}}
\label{app:timed:synthesis}

\subsection{Zero-starting winning conditions}
\label{sec:zero:starting}

A timed language $W \subseteq (\Sigma \cdot \Rpos)^\omega$
is \emph{zero-starting} iff all its words
$\tuple{a_0, t_0} \tuple{a_1, t_1} \cdots \in W$
satisfy $t_0 = 0$.
We show that solving an arbitrary timed game
reduces to solving one with a zero-starting winning condition.
Let $G = \game A B W$ be a timed game,
where $W \subseteq (A \cdot B \cdot \Rpos)^\omega$.
We design an equivalent timed game $G' = \game {A'} B {W'}$,
where actions of Player \I are in $A' = A \cup \set{\startletter}$,
and the zero-starting winning condition is $W' = \setof{(\startletter, b, 0) \cdot w}{w\in W}$.
There is a winning $k,m$-controller $\M$ for $G$ if, and only if,
there is a winning $k,m$-controller $\M'$ in $G'$.
Indeed, $\M'$ is obtained from $\M$ by responding arbitrarily to every $\startletter$,
and conversely, $\M$ is obtained from $\M' = \tuple{A, B, \L', \ell_0', \delta'}$
by restricting to $A$ and letting the initial location be the unique $\ell_0$
\st $\delta'(\ell'_0, \startletter, \zeroregion) = \tuple{\ell_0, \_, \_}$.

\subsection{Strictly monotonic winning conditions}
\label{sec:strictly:monotonic}

Solving a timed game $G = \game A B W$ with a monotonic winning condition
$W \subseteq (A \cdot B \cdot \Rpos)^\omega$ 
reduces to solving one $G' = \game {A'} B {W'}$
with a strictly monotonic winning condition $W' \subseteq (A' \cdot B \cdot \Rpos)^\omega$.
We take Player \I's action to be in $A' = A \cdot \set{0, 1}$.
Consider the function $\phi$ mapping a play in $G'$ of the form
\begin{align}
    \label{eq:strict:pi'}
    \pi' =
        \tuple{\tuple{a_0, f_0}, b_0, t_0'}
            \tuple{\tuple{a_1, f_1}, b_1, t_1'}
                \cdots \in (A' \cdot B \cdot \Rpos)^\omega
\end{align}
to a corresponding play in $G$
\begin{align}
    \label{eq:strict:pi}
    \pi = \phi(\pi') = \tuple{a_0, b_0, t_0} \tuple{a_1, b_1, t_1} \cdots \in (A \cdot B \cdot \Rpos)^\omega
\end{align}
where the new sequence of timestamps $t_0t_1\cdots \in \Rpos^\omega$ is defined as
$t_0 = t_0'$ and, inductively, $t_{i+1} = t_i$ if $f_{i+1} = 0$,
and $t_{i+1} = t_{i+1}'$ otherwise.
Let $W_< = \setof{\pi'}{t_0' < t_1' < \cdots}$ be the language of strictly monotonic plays.
The winning condition in $G'$ is then $$W' = \phi^{-1}(W) \cap W_<.$$
We argue that the two games have the same winner.
\begin{lemma}
    If Player \II has a $k, m$-winning controller in $G$,
    then the same holds in $G'$.
\end{lemma}
\begin{proof}
    Let $\M = \tuple{A, B, \L, \ell_0, \delta}$ be a $k, m$-winning controller for Player \II in $G$.
    We build a winning controller $\M' = \tuple{A', B, \L', \ell'_0, \delta'}$
    for the same player in $G'$ as follows.
    Control locations are $\L' = \L \cdot \regions k m$,
    the initial location is $\ell'_0 = \tuple{\ell_0, \zeroregion}$,
    and the transition relation $\delta'$ is defined,
    for every input $\tuple{\ell, \varphi}, \tuple{a, f}, \varphi'$, as
    \begin{align*}
        \delta'(\tuple{\ell, \varphi}, \tuple{a, f}, \varphi') = 
        \left\{\begin{array}{ll}
            \tuple{\tuple{\ell', \varphi}, b, \Y}
                &\text{ if $f = 0$ and $\delta(\ell, a, \varphi) = \tuple{\ell', b, \Y}$,} \\
            \tuple{\tuple{\ell', \varphi'}, b, \Y}
                &\text{ if $f = 1$ and $\delta(\ell, a, \varphi') = \tuple{\ell', b, \Y}$.}
        \end{array}\right.
    \end{align*}
    Assume $\pi'$ is an $\M'$-conform play as in \eqref{eq:strict:pi'}.
    If it is not strictly monotonic, then $\pi' \not\in W'$ and we are done.
    Otherwise, assume $\pi'$ is strictly monotonic.
    Towards reaching a contradiction,
    assume $\pi' \in \phi^{-1}(W)$.
    Therefore, $\pi = \phi(\pi') \in W$ as in \eqref{eq:strict:pi}.
    By the definition of $\delta'$,
    $\pi$ is $\M$-conform, contradicting that $\M$ is winning.
\end{proof}

\begin{lemma}
    If Player \II has a $k, m$-winning controller in $G'$, then the same holds in $G$.
\end{lemma}
\begin{proof}
    Let $\M' = \tuple{A', B, \L', \ell'_0, \delta'}$ be a $k, m$-winning controller for Player \II in $G'$.
    We assume \wlg that Player \II remembers the input region when the flag $f = 1$ was played last.
    Thus, locations in $L'$ are of the form $\tuple {\ell, \varphi}$.
    We build a winning controller $\M = \tuple{A, B, \L', \ell'_0, \delta}$ for Player \II in $G$
    where 
    \begin{align*}
        \delta(\tuple{\ell, \varphi}, a, \varphi') = 
        \left\{\begin{array}{ll}
            \tuple{\tuple{\ell', \varphi}, b, \Y}
                &\text{ if $\varphi' = \varphi$ and }
                    \delta'(\tuple{\ell, \varphi}, \tuple{a, 0}, \varphi') = \tuple{\tuple{\ell', \varphi}, b, \Y}, \\
            \tuple{\tuple{\ell', \varphi'}, b, \Y}
                &\text{ if $\varphi' \neq \varphi$ and }
                    \delta'(\tuple{\ell, \varphi}, \tuple{a, 1}, \varphi') = \tuple{\tuple{\ell', \varphi'}, b, \Y}.
        \end{array}\right.
    \end{align*}
    Let $\pi$ be a $\M$-conform play and assume towards a contradiction that $\pi \in W$.
    We can chose sufficiently small increments
    in order to make all sequences of equal timestamps in $\pi$ become strictly monotonic,
    and choose the flags $f_i$ accordingly,
    and obtain a play $\pi'$ \st $\pi = \phi(\pi')$.
    By the definition of $\delta$, $\pi'$ is $\M'$-conform.
    But $\pi' \in W'$, contradicting that $\M'$ is winning in $G'$.
\end{proof}

\subsection{Proof of \Cref{lem:Second:Impl}}
\label{app:Second:Impl}

\lemSecondImpl*

\para{Complete winning controllers}

In what follows we restrict to plays satisfying $W^\I_k$.
For proving \Cref{lem:Second:Impl}, the converse of \cref{lem:km:untimed},
we need to understand the general shape
of any possible untimed winning controller
$\M' = \tuple{A', B', \L', \ell'_0, \delta'}$ in $G'$.
We say that such an $\M'$ is \emph{complete}
if its control locations are of the form
$\L' = \L \cdot \regions \X m \cdot \fregions \X$,
$\ell'_0 = \tuple{\ell_0, \zeroregion, \f_0}$,
and every $\M'$-conform run is of the form
\begin{align}
    \label{eq:complete:run}
        \tuple{\ell_0, \r_0, \f_0}
        \tuple{(a_1, \f_1), (b_1, \Y_1), t_1, \tuple{\ell_1, \r_1, \f'_1}}
        \tuple{(a_2, \f_2), (b_2, \Y_2), t_2, \tuple{\ell_2, \r_2, \f'_2}} \cdots, 
\end{align}
where for each $i\geq 1$, the fractional region $\f'_i$ stored in a location agrees with the region $\r_i$, 
its domain $\dom {\f'_i} = \setof{\x\in\X}{\text{there is an $\x$-request at time $u$ with $t_i - 1 < u \leq t_i$}}$, 
and
$\r_i = \region \X m {\mu_i}$ for the clock valuations $\mu_i$ as defined in~\eqref{eq:mudef}.
%
%
It is not difficult to see that complete winning controllers suffice in $G'$.

\begin{restatable}{lemma}{lemComplete}
    If there is a winning untimed controller $\M'$ in $G'$,
    then there is a winning complete one.
\end{restatable}
%
%
\begin{proof}
    When Player \I plays $a' = \tuple{a, \f}$, 
    the complete controller simulates $\M'$.
    Additionally, it uses  the fractional region $\f$ and current region $\r$
    to compute the next region $\r'$ (similarly as in the proof of
    \Cref{lem:km:untimed})
    and the next fractional region $\f'$.
    Let $\hat \r = \xsuccessor {\X, m} \r {\f}$, hence $\f$ agrees with $\hat \r$.
    Then, $\r' = \hat \r$  in improper moves,
    and in proper moves of the form $b' = \tuple{b, \Y}$,
    let $\r' = \extend \Y 0 {\hat \r}$.
    Let $\f''$ be restriction of $\f$ to $\dom \f \setminus \one{\f}$, and let
    $\dom {\f'} = \dom {\f''} \cup \Y$ and
    $\f' = \extend \Y 0 {\f''}$ (thus $\dom{\f'}$ possibly increases
    in the case of proper move).
    %
    This ensures that $\f'$ agrees with $\r'$ and $\dom {\f'}$ contains all requested clocks.
\end{proof}

\begin{lemma} \label{lem:untimed:km}
    If there is a complete winning untimed controller $\M'$ in $G'$ then
    there is a winning $k, m$-controller $\M$ in $G$.
\end{lemma}
\begin{proof}
    Let $\M' = \tuple{A', B', \L', \ell'_0, \delta'}$
    be a winning complete controller in $G'$
    with $\L' = \L \cdot \regions \X m \cdot \fregions \X$,
    $\ell'_0 = \tuple{\ell_0, \r_0, \f_0}$,
    and update function of the form
    $\delta' : \L' \cdot A' \to \L' \cdot B'$.
    We define a winning $k, m$-controller
    $\M = \tuple{A, B, \L', \ell_0', \delta}$ in $G$
    over the same set of control locations $\L'$,
    and update function $\delta : \L' \cdot A \cdot \regions \X m \to \L' \cdot B \cdot 2^\X$.
    In order to define one step of $\delta$ (which corresponds to a proper move)
    we need to take many steps of $\delta'$
    to skip all improper moves preceding the corresponding proper one.
    Let 
    \begin{align} \label{eq:deltatodef}
        \delta(\tuple{\ell, \r, \f}, a, \hat\r) = \tuple{\tuple{\ell'', \r'', \f''}, b, \Y},
    \end{align}
    for $a\in A$, be recursively defined as follows:
    \begin{enumerate}

        \item In the base case, we have $\r \preceq \hat\r$ and $\f$ agrees with $\hat\r$ (as a special case we may have $\r = \hat \r$).
        We apply the transition function of $\M'$ and obtain directly the r.h.s.~in \eqref{eq:deltatodef} as
        $\tuple{\tuple{\ell'', \r'', \f''}, b, \Y} = \delta'(\tuple{\ell, \r, \f}, \tuple{a, \f})$
        where $\r'' = \extend \Y 0 \r$ and $\f''$ agrees with $\r''$.

        \item In the next case, we have $\r \prec \hat \r$ and $\f$ does not agree with $\hat\r$.
        Let $\f'$ be the immediate successor of $\f$,
        and let
        %
            $\delta'(\tuple{\ell, \r, \f}, \tuple{\tick, \f'})
                = \tuple{\tuple{\ell', \r', \bar\f'}, \tuple{\tick, \_}}$,
        %
        where necessarily $\r' = \xsuccessor {\X, m} \r {\f'}$,
        and $\r'$ agrees with $\f'$.
        Then, we recursively define the r.h.s.~in \eqref{eq:deltatodef} as
        $\tuple{\tuple{\ell'', \r'', \f''}, b, \Y} = \delta(\tuple{\ell', \r', \bar\f'}, a, \hat \r)$.

        \item In any other case,
        $\hat \r$ is not a successor region of $\r$.
        Thanks to completeness \eqref{eq:complete:run},
        $\r$ is the region of the current clock valuation,
        and thus the controller can be defined arbitrarily 
        because Player \II is already winning, 
        since Player \I is losing due to violation of $W^\I_k$.    

    \end{enumerate}
    The recursion above ends, and thus $\delta$ is well-defined,
    since there are only finitely many regions and $\prec$ is a strict total order on regions.

    Consider an infinite $\M$-conform run $\rho \in \iRun \M$.
    By the definition of $\delta$,
    there is a corresponding $\M'$-conform run $\rho' \in \iRun \M$ as in \eqref{eq:complete:run}
    where Player \I in $G'$ plays optimally (satisfying $W^\I_k$),
    and $\rho$ arises from $\rho'$ by combining together adjacent sequences of improper moves:
    Let the proper moves in $\rho'$ be at indices $1 = i_1 < i_2 < \cdots$.
    %
    Then, $\rho$ is of the form
    \begin{align*}
        \rho =
            &\tuple{\tuple{\ell_{0}, \r_{0}, \f_{0}}, \mu_{0}}
            \tuple{a_1, b_1, t_{i_1}, \tuple{\ell_{i_1}, \r_{i_1}, \f_{i_1}}, \mu_{i_1}}
            \tuple{a_2, b_2, t_{i_2}, \tuple{\ell_{i_2}, \r_{i_2}, \f_{i_2}}, \mu_{i_2}}
            \cdots, \text{ where } \\
            &\;a_j = \phi(a_{i_j}') \text{ and }
            b_j = \phi(b_{i_j}').
    \end{align*}
    Since Player \I plays optimally when building $\rho'$,
    the corresponding play
    $\pi' = \rtop {\rho'} = \tuple{a'_1, b'_1, t_1} \tuple{a'_2, b'_2, t_2} \cdots $
    is in $W^\I_k$,
    and since $\M'$ is winning,
    $\pi' \in W^\II_{k, m}$ and $\pi' \notin \phi^{-1}(W)$.
    If the corresponding play
    $\pi = \rtop \rho = \tuple{a_1, b_1, t_{i_1}} \tuple{a_2, b_2, t_{i_2}} \cdots$
    in $G$ was winning for Player \I, which means $\pi \in W$,
    since $\phi(\pi') = \pi$
    we would have $\pi' \in \phi^{-1}(W)$, a contradiction.
\end{proof}

\section{Undecidability of timed synthesis for \kNTA{1} conditions}

In this section we show that the timed synthesis problem is undecidable,
thus complementing the decidability results in \cref{sec:timed:synthesis}
about the $k$-timed synthesis problems
when the number of clocks $k$ available to Player \II is fixed in advance.
We show undecidability already in the case when the winning condition of Player \I is a \kNTA 1 language.

\thmSynUnd*

We reduce from the finiteness problem for lossy counter machines,
which is undecidable \cite[Theorem 13]{Mayr:TCS:2003}.
A \emph{$k$-counter lossy counter machine} (\kLCM k) is a tuple $M = \tuple {C, Q, q_0, \Delta}$,
where $C = \set{c_1, \dots, c_k}$ is a set of $k$ counters,
$Q$ is a finite set of control locations,
$q_0 \in Q$ is the initial control location,
and $\Delta$ is a finite set of instructions of the form $\tuple {p, \op, q}$,
where $\op$ is one of $\incr c$, $\decr c$, and $\ztest c$.
A configuration of an \LCM $M$ is a pair $\tuple {p, u}$,
where $p \in Q$ is a control location,
and $u \in \N^C$ is a counter valuation.
For two counter valuations $u, v \in \N^C$,
we write $u \leq v$ if $u(c) \leq v(c)$ for every counter $c \in C$.
The semantics of an \LCM $M$ is given by a (potentially infinite) transition system over the configurations of $M$
\st there is a transition $\tuple {p, \mu} \goesto \delta \tuple {q, \nu}$,
for
$\delta = \tuple {p, \op, q} \in \Delta$, whenever
%
\begin{enumerate}[label=\textit{\arabic*)}]
    \item $\op = \incr c$ and $\nu \leq \extend c {\mu(c) + 1} \mu$, or
    \item $\op = \decr c$ and $\nu \leq \extend c {\mu(c) - 1} \mu$, or
    \item $\op = \ztest c$ and $\mu(c) = 0$ and $\nu \leq \mu$.
\end{enumerate}
The \emph{finiteness problem} (a.k.a.~space boundedness) for an \LCM $M$
asks to decide whether the reachability set
$$\reachset M = \setof{\tuple{p, \mu}} {\tuple {q_0, \mu_0} \goesto {}^* \tuple{p, \mu}}$$
is finite, where $\mu_0 = \lambda c \st 0$ is the constantly $0$ counter valuation.
\begin{theorem}[\protect{\cite[Theorem 13]{Mayr:TCS:2003}}]
    The \kLCM 4 finiteness problem is undecidable.
\end{theorem}

We use the following encoding of \LCM runs
(c.f.~\cite[Definition 4.6]{LW08} for a similar encoding) into timed words.
We assume that there are four lossy counters $C = \set{c_1, c_2, c_3, c_4}$.
A strictly monotonic timed word $u$ (i.e., any two adjacent letters therein occur one strictly after the other)
over alphabet $C$ whose untiming is of the form
$\untime u = c_1^{n_1} c_2^{n_2} c_3^{n_3} c_4^{n_4}$
encodes the counter valuation $\mu \in \N^C$ defined by
$\mu(c_j) = n_j$ for every $j \in \set{1, 2, 3, 4}$.
In this case, we slightly abuse notation and write $u(c_j) = n_j$.
A timed word $\pi_A$ over alphabet $A = Q \cup \Delta \cup C$
is a correct encoding of an \LCM run
\begin{align*}
    \tuple{p_n, u_n}
        \goesto {\delta_{n-1}} \tuple{p_{n-1}, u_{n-1}}
            \goesto {\delta_{n-2}} \cdots
                \goesto {\delta_0} \tuple{p_0, u_0}
\end{align*}
if its untiming is of the form
%
\begin{align*}
    \untime {\pi_A} = p_0 u_0 \delta_0 \quad \cdots \quad p_{n-1} u_{n-1} \delta_{n-1} \quad p_n u_n
\end{align*}
and the following conditions are satisfied:
\renewcommand{\labelenumii}{(C\arabic{enumi}.\arabic{enumii})}
\renewcommand{\labelenumiii}{(C\arabic{enumi}.\arabic{enumii}.\arabic{enumiii})}

\begin{enumerate}[label=(C\arabic*)]

    \item \label{cond:1}
    for every $i$,
    $p_i \in Q$, $u_i \in \set {c_1}^* \set {c_2}^* \set {c_3}^* \set {c_4}^*$,
    and $\delta_i$ is a transition of the form $\delta_i = \tuple{p_{i+1}, \op, p_i}$;

    \item \label{cond:2} $p_0$ occurs at time 0;


    \item \label{cond:3} for every $0 \leq i < n$, $p_{i+1}$ occurs exactly one time unit after $p_i$;

    \item \label{cond:4} $\pi_A$ is strictly monotonic;


    \item \label{cond:5}
    for every transition $\delta_i = \tuple{p_{i+1}, \op, p_i}$
    and counter $c_j \in \set {c_1, c_2, c_3, c_4}$,
    
    \begin{enumerate}

        \item \label{cond:5:1}
        if $\op = \incr {c_j}$,
        then each occurrence of $c_j$ in $u_i$
        is followed by an occurrence of $c_j$ in $u_{i+1}$ after exactly one time unit,
        perhaps with the exception of the last occurrence of $c_j$ in $u_i$;
        consequently, $u_{i+1}(c_j) \geq u_i(c_j) - 1$.

        \item \label{cond:5:2}
        if $\op = \decr {c_j}$, then
        \begin{enumerate}
            
            \item \label{cond:5:2:1} each occurrence of $c_j$ in $u_i$
            is followed by an occurrence of $c_j$ in $u_{i+1}$
            after exactly one time unit,
            and moreover
            
            \item \label{cond:5:2:2} the last occurrence of $c_j$ in $u_{i+1}$
            does not have a matching occurrence one time unit earlier in $u_i$;
            
        \end{enumerate}
        consequently, $u_{i+1}(c_j) \geq u_i(c_j) + 1$.
            
        \item \label{cond:5:3}
        if $\op = \ztest {c_j}$,
        then $u_{i+1}(c_j) = u_i(c_j) = 0$.

        \item \label{cond:5:4}
        otherwise,
        each occurrence of $c_j$ in $u_i$
        is followed by an occurrence of $c_j$ in $u_{i+1}$ after exactly one time unit;
        consequently, $u_{i+1}(c_j) \geq u_i(c_j)$.

    \end{enumerate}

\end{enumerate}

We design a game where Player \I builds encodings of \LCM runs as above;
accordingly, let her actions be $A$.
Player \II either plays $\OK$ when she believes that the encoding so far does not contain any mistake,
or she will play an action of the form $\ERROR_e$
when she believes that an error of type $e$ occurred (to be explained below),
where
$$e \in \set{1, 2, 3, 4} \cup \set{5.1, 5.2.1, 5.2.2, 5.3, 5.4} \cdot \set{c_1, c_2, c_3, c_4} \cdot \set{T_1, T_2}.$$ 
Let $\pi = a_1 b_1 t_1 \cdots a_i b_i t_i \in (A \cdot B \cdot \Rpos)^*$
be the actions played till the end of round $i$,
and let $\pi_A = a_1 t_1 \cdots a_n t_n \in (A \cdot \Rpos)^*$
be the corresponding purported encoding of (a prefix of) an \LCM run.
Let $a_j$ be the last action of the form $a_j = \delta = \tuple{\_, \op, \_}$.
%
The most common type of error in the encoding is that a $c_j$
does not have a matching occurrence of $c_j$ one time unit later.
There are two possible ways in which such a disappearence may occurr:
\begin{enumerate}
    \item[$T_1$:]
    Letter $c_j$ occurs at time $t = t_i - 1$ and $a_i \neq c_j$.

    \item[$T_2$:]
    Letter $c_j$ occurs at some time $t_{i-1} - 1 < t < t_i - 1$.
\end{enumerate}
We require Player \II to specify precisely which variant $X \in \set{T_1, T_2}$
of the error actually occurred.
It will be convenient to define the predicate $P(c_j, X)$
which holds if Player \II incorrectly marks the disappearance of $c_j$,
i.e., either $X = T_1$ and if there is an earlier occurrence of $c_j$ at time $t = t_i - 1$ then $a_i = c_j$, or
$X = T_2$ and there is an earlier occurrence of $c_j$ at time $t \in \set{t_{i-1}, t_i}$
(both conditions are \kNTA 1-recognisable).
%
%
%
%
We are now ready to define the winning condition of the game.
If Player \II plays $\OK$ but $\pi_A$ contains an error
violating one of the conditions \ref{cond:1}--\ref{cond:5},
then the game ends and Player \I wins immediately.
(Plays of this form can be recognised by a \kNTA 1 as in \cite{LW08}).)
If Player \II plays $\ERROR_e$,
then the game ends and Player \I wins iff an error of type $e$ did not occur.
This is the case if any of the following conditions mimicking \ref{cond:1}--\ref{cond:5} holds:

\begin{enumerate}[label=(W\arabic*),ref=(W\arabic*)]

    \item \label{W:cond:1} Player \II played $b_i = \ERROR_1$
    but \ref{cond:1} is satisfied.

    \item \label{W:cond:2} Player \II played $b_i = \ERROR_2$
    but \ref{cond:2} is satisfied.
    
    \item \label{W:cond:3} Player \II played $b_i = \ERROR_3$
    but \ref{cond:3} is satisfied.

    \item \label{W:cond:4} Player \II played $b_i = \ERROR_4$
    but \ref{cond:4} is satisfied.



    \item \label{W:cond:5} 
    Player \II incorrectly marks that condition \ref{cond:5} is not satisfied:
        \begin{enumerate}[label=(W\arabic{enumi}.\arabic*),ref=(W\arabic{enumi}.\arabic*)]

            \item \label{W:cond:5:1}
            Player \II plays $b_i = \ERROR_{5.1, c_j, X}$ 
            and either $\op \neq \incr {c_j}$,
            or $P(c_j, X)$ holds,
            or there is an occurrence of $c_j$ at some time $t_{i-1} - 1 \leq t \leq t_i - 1$
            which is immediately followed by another occurrence of $c_j$
            (and thus it is not the last one).

            \item \label{W:cond:5:2}
            Payer \II plays $b_i = \ERROR_{5.2.N, c_j, X}$
            and either $\op \neq \decr {c_j}$, or

            \begin{enumerate}[label=(W\arabic{enumi}.\arabic{enumii}.\arabic*),ref=(W\arabic{enumi}.\arabic{enumii}.\arabic*)]
                \item \label{W:cond:5:2:1}
                $N = 1$ and $P(c_j, X)$, or

                \item \label{W:cond:5:2:1}
                $N = 2$ ($X$ is irrelevant in this case) and:
                either $a_i = c_j$ (thus the last occurrence of $c_j$ has possibly not been seen);
                or $a_{i-1} = c_j, a_i \neq c_j$ (thus the last occurrence has been seen),
                but there is an occurrence of $c_j$ at time $t = t_{i-1}-1$
                (this last occurrence has a match one time unit before).
            \end{enumerate}

            \item \label{W:cond:5:3}
            Player \II plays $b_i = \ERROR_{5.3, c_j, X}$ ($X$ is irrelevant in this case)
            and either $\op \neq \ztest {c_j}$,
            or there is no occurrence of $c_j$ in the last two configurations $u_{i-1}, u_i$.

            \item \label{W:cond:5:4}
            Player \II plays $b_i = \ERROR_{5.4, c_j, X}$ 
            and either $\op$ involves counter $c_j$, or $P(c_j, X)$.
        \end{enumerate}
\end{enumerate}
Finally, if the game goes on forever, then Player \I loses.
All conditions \ref{W:cond:1}--\ref{W:cond:5} are \kNTA 1 recognisable
(condition \ref{W:cond:5:3} is even untimed),
and so is their disjunction.
The following lemma states correctness of the reduction.
\begin{lemma}
	The set of reachable configurations $\reachset M$ is finite
	if, and only if,
	there is a winning controller for Player \II in the game.
\end{lemma}

\begin{proof}
    For the ``only if'' direction, assume that $\reachset M$ is finite.
    There is some $k$ s.t.~every reachable configuration $\tuple {p, \mu}$
    has size $\mu(c_1)+\mu(c_2)+\mu(c_3)+\mu(c_4)+1 \leq k$.
    In this case, the set of correct timed encodings of runs of $M$
    can be recognised by a \kDTA {(k+2)} $A$
    which resets clock $\x_j$ when reading the $j$-th position of block $p_i u_i \delta_i$
    (which is of length $\leq k + 2$).
    From $A$ we can immediately produce a winning controller for Player \II with $k$ clocks:
    The controller reads the word and checks membership in $\lang A$,
    outputting $\OK$ when membership holds and the appropriate error $\ERROR_e$ otherwise.
    The exact error $e$ can deterministically be determined
    by looking at the values of the clocks $\x_1, \dots, \x_{k+2}$ (details omitted).

    For the ``if'' direction, assume that $\reachset M$ is infinite,
    and thus there exist reachable configurations with arbitrarily large counter values.
    Suppose, towards reaching a contradiction,
    that Player \II has a winning controller $\M$ with $k$ clocks.
    We can see $\M$ as a \kDTA k which additionally produces at each step
    an action of the form $\OK$ or $\ERROR_e$
    (in a deterministic manner, just based on the current input and state).
    We can produce a \kDTA k $A$ by removing all transition outputting actions of the form $\ERROR_e$,
    remove the output labelling $\OK$ from the remaining transitions,
    and make all the remaining reachable control locations accepting.
    Since $\M$ is winning, it outputs $\OK$ precisely when the encoding is correct.
    Therefore, the $A$ just constructed recognises precisely the set of correct encodings of runs of $M$.
    We show that this leads to a contradiction, using the fact that $M$ is unbounded.
    There exists a run $\pi$ of $M$ where some counter value exceeds $k$,
    and thus when $A$ reads the reversal-encoding of $\pi$
    it must forget some timestamp (say) $\tuple {c_1, t}$
    from configuration $p_i \delta_i u_i$.
    Since $t$ is forgotten, we can perturb its corresponding $\tuple {c_1, t+1}$
    in $p_{i+1} \delta_{i+1} u_{i+1}$
    to any value $\tuple {c_1, t'}$ \st $t' - t \neq 1$
    and obtain a new word still accepted by $A$,
    but which is no longer the reversal-encoding of a run of $M$,
    thus reaching the sought contradiction.
\end{proof}

\end{document}